\documentclass[pra,aps,nopacs,onecolumn,twoside,superscriptaddress]{revtex4}

\usepackage{soul} 
\usepackage{multirow}
\usepackage{amsmath,amsfonts,amssymb,caption,color,epsfig,graphics,graphicx,hyperref,latexsym,mathrsfs,revsymb,theorem,url,verbatim,epstopdf,enumerate}
\usepackage{amsmath,amsfonts,amssymb,caption,hyperref,color,epsfig,graphics,graphicx,latexsym,mathrsfs,revsymb,theorem,url,verbatim,epstopdf,cleveref}
\usepackage{fontenc}
\hypersetup{colorlinks,linkcolor={blue},citecolor={blue},urlcolor={red}}

\usepackage{cleveref}

\newtheorem{definition}{Definition}
\newtheorem{proposition}[definition]{Proposition}
\newtheorem{lemma}[definition]{Lemma}

\newtheorem{theorem}[definition]{Theorem}
\newtheorem{corollary}[definition]{Corollary}
\newtheorem{conjecture}[definition]{Conjecture}

\newtheorem{remark}[definition]{Remark}
\newtheorem{example}[definition]{Example}
\newtheorem{question}[definition]{Question}
\newtheorem{memo}[definition]{Memo}

\def\squareforqed{\hbox{\rlap{$\sqcap$}$\sqcup$}}
\def\qed{\ifmmode\squareforqed\else{\unskip\nobreak\hfil
\penalty50\hskip1em\null\nobreak\hfil\squareforqed
\parfillskip=0pt\finalhyphendemerits=0\endgraf}\fi}
\def\endenv{\ifmmode\;\else{\unskip\nobreak\hfil
\penalty50\hskip1em\null\nobreak\hfil\;
\parfillskip=0pt\finalhyphendemerits=0\endgraf}\fi}
\newenvironment{proof}{\noindent \textbf{{Proof.~} }}{\qed}
\def\Dbar{\leavevmode\lower.6ex\hbox to 0pt
{\hskip-.23ex\accent"16\hss}D}
\makeatletter
\def\url@leostyle{%
  \@ifundefined{selectfont}{\def\UrlFont{\sf}}{\def\UrlFont{\small\ttfamily}}}
\makeatother
\def\bcj{\begin{conjecture}}
\def\ecj{\end{conjecture}}
\def\bcr{\begin{corollary}}
\def\ecr{\end{corollary}}
\def\bd{\begin{definition}}
\def\ed{\end{definition}}
\def\bea{\begin{eqnarray}}
\def\eea{\end{eqnarray}}
\def\bem{\begin{enumerate}}
\def\eem{\end{enumerate}}
\def\bex{\begin{example}}
\def\eex{\end{example}}
\def\bim{\begin{itemize}}
\def\eim{\end{itemize}}
\def\bl{\begin{lemma}}
\def\el{\end{lemma}}
\def\bma{\begin{bmatrix}}
\def\ema{\end{bmatrix}}
\def\bpf{\begin{proof}}
\def\epf{\end{proof}}
\def\bpp{\begin{proposition}}
\def\epp{\end{proposition}}
\def\bqu{\begin{question}}
\def\equ{\end{question}}
\def\br{\begin{remark}}
\def\er{\end{remark}}
\def\bt{\begin{theorem}}
\def\et{\end{theorem}}
\def\bmm{\begin{memo}}
\def\emm{\end{memo}}

\def\btb{\begin{tabular}}
\def\etb{\end{tabular}}

\newcommand{\nc}{\newcommand}

\nc{\as}{{\cal AS}}
\nc{\app}{{\cal AP}}

\def\a{\alpha}
\def\b{\beta}
\def\g{\gamma}
\def\d{\delta}

\def\l{\lambda}

\def\r{\rho}

\def\ps{\psi}

\def\G{\Gamma}

\def\O{\Omega}

 \nc{\bbA}{\mathbb{A}} \nc{\bbB}{\mathbb{B}} \nc{\bbC}{\mathbb{C}}
 \nc{\bbD}{\mathbb{D}} \nc{\bbE}{\mathbb{E}} \nc{\bbF}{\mathbb{F}}
 \nc{\bbG}{\mathbb{G}} \nc{\bbH}{\mathbb{H}} \nc{\bbI}{\mathbb{I}}
 \nc{\bbJ}{\mathbb{J}} \nc{\bbK}{\mathbb{K}} \nc{\bbL}{\mathbb{L}}
 \nc{\bbM}{\mathbb{M}} \nc{\bbN}{\mathbb{N}} \nc{\bbO}{\mathbb{O}}
 \nc{\bbP}{\mathbb{P}} \nc{\bbQ}{\mathbb{Q}} \nc{\bbR}{\mathbb{R}}
 \nc{\bbS}{\mathbb{S}} \nc{\bbT}{\mathbb{T}} \nc{\bbU}{\mathbb{U}}
 \nc{\bbV}{\mathbb{V}} \nc{\bbW}{\mathbb{W}} \nc{\bbX}{\mathbb{X}}
 \nc{\bbZ}{\mathbb{Z}}

 \nc{\bA}{{\bf A}} \nc{\bB}{{\bf B}} \nc{\bC}{{\bf C}}
 \nc{\bD}{{\bf D}} \nc{\bE}{{\bf E}} \nc{\bF}{{\bf F}}
 \nc{\bG}{{\bf G}} \nc{\bH}{{\bf H}} \nc{\bI}{{\bf I}}
 \nc{\bJ}{{\bf J}} \nc{\bK}{{\bf K}} \nc{\bL}{{\bf L}}
 \nc{\bM}{{\bf M}} \nc{\bN}{{\bf N}} \nc{\bO}{{\bf O}}
 \nc{\bP}{{\bf P}} \nc{\bQ}{{\bf Q}} \nc{\bR}{{\bf R}}
 \nc{\bS}{{\bf S}} \nc{\bT}{{\bf T}} \nc{\bU}{{\bf U}}
 \nc{\bV}{{\bf V}} \nc{\bW}{{\bf W}} \nc{\bX}{{\bf X}}
 \nc{\bZ}{{\bf Z}}

\nc{\cA}{{\cal A}} \nc{\cB}{{\cal B}} \nc{\cC}{{\cal C}}
\nc{\cD}{{\cal D}} \nc{\cE}{{\cal E}} \nc{\cF}{{\cal F}}
\nc{\cG}{{\cal G}} \nc{\cH}{{\cal H}} \nc{\cI}{{\cal I}}
\nc{\cJ}{{\cal J}} \nc{\cK}{{\cal K}} \nc{\cL}{{\cal L}}
\nc{\cM}{{\cal M}} \nc{\cN}{{\cal N}} \nc{\cO}{{\cal O}}
\nc{\cP}{{\cal P}} \nc{\cQ}{{\cal Q}} \nc{\cR}{{\cal R}}
\nc{\cS}{{\cal S}} \nc{\cT}{{\cal T}} \nc{\cU}{{\cal U}}
\nc{\cV}{{\cal V}} \nc{\cW}{{\cal W}} \nc{\cX}{{\cal X}}
\nc{\cZ}{{\cal Z}}

\nc{\hA}{{\hat{A}}} \nc{\hB}{{\hat{B}}} \nc{\hC}{{\hat{C}}}
\nc{\hD}{{\hat{D}}} \nc{\hE}{{\hat{E}}} \nc{\hF}{{\hat{F}}}
\nc{\hG}{{\hat{G}}} \nc{\hH}{{\hat{H}}} \nc{\hI}{{\hat{I}}}
\nc{\hJ}{{\hat{J}}} \nc{\hK}{{\hat{K}}} \nc{\hL}{{\hat{L}}}
\nc{\hM}{{\hat{M}}} \nc{\hN}{{\hat{N}}} \nc{\hO}{{\hat{O}}}
\nc{\hP}{{\hat{P}}} \nc{\hR}{{\hat{R}}} \nc{\hS}{{\hat{S}}}
\nc{\hT}{{\hat{T}}} \nc{\hU}{{\hat{U}}} \nc{\hV}{{\hat{V}}}
\nc{\hW}{{\hat{W}}} \nc{\hX}{{\hat{X}}} \nc{\hZ}{{\hat{Z}}}

\nc{\hn}{{\hat{n}}}

\def\diag{\mathop{\rm diag}}
\def\dim{\mathop{\rm Dim}}

\def\ghz{\mathop{\rm GHZ}}

\def\max{\mathop{\rm max}}
\def\min{\mathop{\rm min}}

\def\tr{\mathop{\rm Tr}}

\def\ra{\rightarrow}

\newcommand{\bra}[1]{\langle#1|}
\newcommand{\ket}[1]{|#1\rangle}
\newcommand{\proj}[1]{| #1\rangle\!\langle #1 |}

\newcommand{\abs}[1]{|#1|}

\def\Dbar{\leavevmode\lower.6ex\hbox to 0pt
{\hskip-.23ex\accent"16\hss}D}

\begin{document}

\title{Eigenvalues of multipartite entanglement witnesses}

\date{\today}

\author{Nalan Wang}\email[]{nalanwang@buaa.edu.cn}
\affiliation{LMIB(Beihang University), Ministry of Education, and School of Mathematical Sciences, Beihang University, Beijing 100191, China}

\author{Lin Chen}\email[]{linchen@buaa.edu.cn (corresponding author)}
\affiliation{LMIB(Beihang University), Ministry of Education, and School of Mathematical Sciences, Beihang University, Beijing 100191, China}

\pacs{03.65.Ud, 03.67.Mn}


\begin{abstract}
We investigate various properties of multipartite block-positive operators, decomposable EWs (DEWs), and non-decomposable EWs. We provide a necessary and sufficient condition to construct a special DEW using the multipartite GHZ state. For multipartite DEW, we explicitly characterize the supremum and infimum of maximum and minimum eigenvalues, as well as the trace of its square. We also derive other results concerning NDEW, eigenvalues of $2 \times n$ EWs and the corresponding physical implications. Furthermore, we investigate the tightness of inequalities of these eigenvalues with examples. 
\end{abstract}


\maketitle


Keywords: eigenvalue, entanglement witness, multipartite system 


\section{Introduction}

Detecting entanglement is a basic task in quantum-information processing such as computing and cryptography. A popular method is to construct entanglement witnesses (EWs), which has received extensive attentions in both theory and experiment in past decades, for a review see \cite{Chru2014Entanglement}. Experimental detection of multipartite entanglement using witness operators have been proposed
\cite{bourennane2004experimental}. The relationship between entanglement witnesses and Bell inequalities was revealed \cite{hyllus2005relations}. The concept of nonlinear EWs was introduced and compared with traditional linear EWs \cite{guhne2006nonlinear,2006FGS,2006OGN}. Device-independent witnesses of genuine multipartite entanglement were developed \cite{2011JDB}. Witnessing quantum coherence was extended from solid-state to biological systems \cite{2012CM}. Multipartite entanglement witnesses were proposed \cite{2013Sperling}. Witnesses of mixed separable states useful for entanglement creation were studied \cite{ganguly2014witness}. The design and experimental performance of local entanglement witness operators were discussed \cite{2019Design}. A measurement-device-independent entanglement witness for tripartite entangled states was proposed and its applications were demonstrated \cite{2020Measurement}. 

From a more mathematical point of view, optimized EWs have been constructed in \cite{lewenstein2000optimization1}. They have been used to characterize separable states
\cite{lewenstein2001characterization}. The optimization of ultrafine entanglement witnesses was explored \cite{2018Optimization}. Next, the concept of mirrored entanglement witnesses was introduced \cite{bae2020mirrored}. The structure of mirrored operators derived from optimal entanglement witnesses was investigated  \cite{bera2023structure}. The existence of a mirrored pair of optimal non-decomposable entanglement witnesses for two qudits was demonstrated \cite{chruscinski2025mirroredn}. Absolute-separability EWs for symmetric multiqubit states were proposed \cite{serrano2024absolute}. Third, the spectral properties of EW were analyzed in \cite{sarbicki2008spectral}. This is related to bound entanglement \cite{chruscinski2009spectral}.
The inverse eigenvalue problem for entanglement witnesses was solved \cite{johnston2018inverse}. The inertias of bipartite EWs have also been studied  \cite{2020Inertias,shen2020inertias,3x3inertia2022changchun,feng2024inertia}. Recently, the supremum and infimum of maximum and minimum eigenvalues of bipartite EWs were studied by analytically constructing novel EWs and mathematical properties \cite{2025Spectral}. Nevertheless, these facts concern little about multipartite EWs. This is the focus of this paper. Actually, the latest research on quantum computing system has achieved $98$-qubit system \cite{Ransford_2026}. 

In this paper, we begin by presenting properties of block-positive matrices in Lemma \ref{le:BPmn} and Remark \ref{rem:equivalentFormBP}. In Lemma \ref{le:ConvexityOfTr(W^2)}, we show the strict convexity of matrix function $\tr(W^2)$ of Hermitian matrix $W$, which paves the way for the subsequent induction-based treatment of multipartite states. In Lemma \ref{lem:Existence}, we  introduce some facts from algebraic geometry. We also review necessary and sufficient conditions detecting entanglement and construction of decomposable EWs(DEWs). Lemmas \ref{le:EW=general}-\ref{le:existence_0f_10000} compile some known results concerning EWs, DEWs, and non-DEWs (NDEWs). Next, Lemmas  \ref{le:multi-AS} and \ref{le:m-partite=sep} present separability and absolute-separability criterion of multipartite states. Theorem \ref{lem:finer_P}-Lemma \ref{le:2^n}, together with Corollary \ref{coro:d1...dn}, present properties of block-positive (BP) matrices, optimality and system split and eigenvalues of multiparite EWs.

In Theorem \ref{le:2tin1234} (i)-(iv), we elucidate the interrelations among the eigenvalues of $2 \times n$ EWs. These results are mainly summarized in Table \ref{tab:LambdaResults}. The "limiting of parameters" column in Table \ref{tab:LambdaResults} shows how the parameters $p$ and $m$, obtained from Theorem \ref{le:2tin1234}, relate to previously known results. In Lemma \ref{le:2n13>=}, we discuss other inequality governing eigenvalues, and further investigated the tightness of these inequalities with examples. In Lemma \ref{le:GHZ_n}, we provide a necessary and sufficient condition to construct a DEW using the multipartite GHZ state. Through an in-depth investigation of two-qubit states, we attempt to generalize some properties of two-qubit states to arbitrary $n$-qubit states. From Theorem \ref{le:d1-dnTrW^2} (i), we obtain the conclusion regarding the infimum of $\tr(W^2)$ of DEW in Table \ref{tab:multipartiteResults}. The supremum of $\tr (W^2)$ of DEW in Table \ref{tab:multipartiteResults} is obtained. The conclusion regarding the supremum of minimum eigenvalue of DEW, i.e., $\l_{d_1 \cdots d_n}$ is from Lemma \ref{le:multi-DEW=lambda1<=1}, and the supremum of maximum eigenvalue of DEW, i.e., $\l_1$ is from Theorem \ref{le:multi-DEW=lambda1<1}. From Theorem \ref{le:multi-DEW=lambda1<1}, we also obtain the infima of the minimum and maximum eigenvalues of DEW. All the results presented in Table \ref{tab:multipartiteResults} are comparable with the corresponding results for two-qubit systems, from which clear regularities can be observed. For instance, one may ask whether these exact values of DEW can indeed be attainable under the same circumstances. Moreover, in the expression for the infimum of the maximum eigenvalue, the parameter $mn$ is replaced by $d_1d_2 \cdots d_n$.  Theorem \ref{le:TrW^2_NDEW} shows that if $W$ is an NDEW, then the infimum of $\tr (W^2)$ is not attainable.



\begin{table}[htbp]
    \centering
    \setlength{\tabcolsep}{1pt} 
    
    \caption{Results on eigenvalues of $2 \times n$ EWs}
    \label{tab:LambdaResults}
    \renewcommand{\arraystretch}{3} 
    
    \begin{tabular}{|l|c|c|c|c|}
        \hline
         & range & limiting of parameters & origin \\ 
        \hline
        
        $\lambda_3$ & $\left[-\frac{1}{(2+2\sqrt{2})(2n-2)}, \frac{1}{\sqrt{3}}\right)$ & / & Theorem $\ref{le:2tin1234}$ (i) \\ 
        \hline
        
        $\lambda_4$ & $\left[-\frac{1}{2(2n-3)}, \frac{1}{2}\right)$ & / & Theorem $\ref{le:2tin1234}$ (i) \\ 
        \hline
        
        $\lambda_k$ & $\left[-\frac{1}{2(2n-k+1)}, \frac{1}{\sqrt{k}}\right)$ & $k=4,5, \cdots 2n$ & Theorem $\ref{le:2tin1234}$ (i) \\ 
        \hline

        $\sum_{i=3}^{2n} \l_{i}$ & $\left[-\frac{1}{2+2\sqrt{2}}, 1\right)$ & $p \rightarrow 1+\sqrt{2}$ & Lemma $\ref{le:EW=general}$ (iv) and Theorem $\ref{le:2tin1234}$ (iii) \\ 
        \hline

        $\sum_{i=4}^{2n} \l_{i}$ & $\left[-\frac{1}{2}, 1\right)$ & $p \rightarrow 1$ or $m \rightarrow 1$ & Lemma $\ref{le:EW=general}$ (iv) and Theorem $\ref{le:2tin1234}$ (iii) \\ 
        \hline

        $2p\sum_{i=4}^{2n} \l_{i} + (p^2-1)\l_3$ & $\left[-1, 1\right]$ & $1 \le p \le 1+\sqrt{2}$ & Theorem $\ref{le:2tin1234}$ (iii) \\ 
        \hline

        $(m^2+2m-1)\sum_{i=4}^{2n} \l_{i} + (m^2-1)(\l_3 + \l_2)$ & $\left[-1, \infty\right)$ & $m \geq 1$ & Theorem $\ref{le:2tin1234}$ (iv) \\ 
        \hline
    \end{tabular}
\end{table}

\begin{table}[htbp]
    \centering
    \setlength{\tabcolsep}{20pt} 
    
    \caption{Results on eigenvalues of multipartite EWs}
    \label{tab:multipartiteResults}
    \renewcommand{\arraystretch}{3} 
    
    \begin{tabular}{|l|c|c|c|}
        \hline
         & DEW,sup & DEW,inf & origin \\ 
        \hline

        $\mathrm{Tr}(W^2)$ & 1, attainable & $\frac{1}{d_1 \cdots d_n - 1}$, not attainable & Theorem $\ref{le:d1-dnTrW^2}$ (i), (iii) \\ 
        \hline
        
        $\lambda_{d_1 \cdots d_n}$ & 0, not attainable & $-\frac{1}{2}$, attainable & Lemma $\ref{le:multi-DEW=lambda1<=1}$, Theorem $\ref{le:multi-DEW=lambda1<1}$ (iv) \\ 
        \hline

        $\lambda_1$ & 1, not attainable & $\frac{1}{d_1 \cdots d_n - 1}$, not attainable & Theorem $\ref{le:multi-DEW=lambda1<1}$ (ii)-(iii) \\ 
        \hline
    \end{tabular}
\end{table}

The detection of entanglement is equivalent to the long-standing separability problem, which is an NP-hard problem in computational complexity. It is known that a separable state is 
positive partial transpose (PPT) \cite{peres1996separability,horodecki1996necessary}. The converse also holds for two-qubit and qubit-qutrit systems, while in higher dimensions there exist the so-called PPT entangled states \cite{horodecki1997separability}. Some PPT states have  been studied numerically \cite{leinaas2007extreme,leinaas2010numerical}. The separability of quantum states of two-qubit and three-qubit states have been studied from the aspect of eigenvalues \cite{slater2009eigenvalues,tanaka2014determining,han2017separability}. Hence, our results may be related to these results for a better understanding between the relation of multipartite PPT entangled states and EWs. 

Next, we provide a series of physical interpretations for the two tables above. For Table \ref{tab:LambdaResults}, we can examine the spectral structure and dimensionality effects of the $2 \times n$ system. From the range of $\lambda_k$ in the third row of Table \ref{tab:LambdaResults}, we take the maximum value $k = 2n$ which corresponds to the smallest eigenvalue, i.e., the ground-state energy in physical terms. One can observe a depth limit on the energy potential well in \cite{nielsen2010quantum}. Namely, no matter how large the subsystem size $n$ is (even as $n \to \infty$), as a legitimate entanglement detection operator, its lowest energy level (ground-state energy) cannot go below $-\frac{1}{2}$ (under the normalized condition), which physically corresponds to the energy threshold in \cite{Igloi2023Entanglement, Zhang2018Characterization}. The spectral structure shown in Table \ref{tab:LambdaResults} directly reflects the energy stability of the system. When parameter variations cause drastic changes in these eigenvalues, this corresponds to the quantum phase transition mentioned in \cite{Pezze2017MultipartiteEI}. We now turn to Table \ref{tab:multipartiteResults} and consider the threshold tool for detecting entanglement, while \cite{Zhang2018Characterization} demonstrates the use of similar mathematical tools to probe higher-order topological properties. Our results, particularly the infimum of eigenvalues of DEWs (e.g., $-\frac{1}{2}$) derived in Table \ref{tab:multipartiteResults}, provide a rigorous energy threshold for detecting entanglement in complex systems experimentally. As demonstrated in their characterization of ultrafast energy-time entangled photon pairs in \cite{MacLean2018Direct}, utilizing Hamiltonian-based entanglement witnesses is crucial for probing quantum correlations on ultrafast time scales.

The rest of this paper is organized as follows. In Sec. \ref{sec:pre} we introduce the main technique and notions used in this paper. In Sec. \ref{sec:res} we show our main results. Finally we conclude in Sec. \ref{sec:con}. To keep the main text concise and conclusions clear, we place most of the proofs in the appendix \ref{app}.

\section{Preliminaries}
\label{sec:pre}

In this section, we introduce the preliminary notions and facts used in this paper. They are split into three parts on linear algebra, quantum information on bipartite and multipartite systems in Sec. \ref{sec:linear algebra}, \ref{sec:qi=bipartite}, and \ref{sec:qi=multipartite}, respectively. We first list some known conclusions, followed by new proofs, most of which can be found in the appendix \ref{app}. Then we further extend them to obtain new conclusions, which serve as a foundation for the subsequent results in Sec. \ref{sec:res}.

\subsection{linear algebra} 
\label{sec:linear algebra}

\begin{definition}
\label{def:mnblock_po}
     A bipartite Hermitian matrix $W \in M_m(\mathbb{C}) \otimes M_n(\mathbb{C})$ is called an $m\times n$ block-positive matrix if  $\langle a, b | W | a, b \rangle \geq 0, \ \forall |a \rangle \in \mathbb{C}^m, \ |b \rangle \in \mathbb{C}^n$. Then we denote the set of these matrices by $\mathcal{BP}_{m,n}$. 
    \qed
\end{definition}

\begin{lemma}
\label{le:BPmn}
    Let $W = [W_{i,j}]_{i,j=1}^{m,n} \in \mathcal{BP}_{m,n}$.

    (i) If $W_{k,k} = O$, then $W_{k,j} = W_{j,k} = O$.

    (ii) If for any $i$, the $k$-th diagonal entry of $W_{i,i}$ vanishes, then the entire $k$-th row and column of every block $W_{i,j}$ must be zero.
    \qed
\end{lemma}

We present a novel proof of Lemma \ref{le:BPmn} in Appendix \ref{app:leBPmn}.

\begin{remark}
\label{rem:equivalentFormBP}
An equivalent form of the above lemma is that 
\begin{eqnarray}
    \mathcal{R}(W) \subseteq \mathcal{R}(W_A) \otimes \mathcal{R}(W_B),
\end{eqnarray}
where $W_A := [\tr(W_{i,j})]_{i,j=1}^{m} = \left[\begin{matrix}\tr W_{1,1}&\cdots&\tr W_{1,m}\\\vdots&\ddots&\vdots\\\tr W_{m,1}&\cdots&\tr W_{m,m}\end{matrix}\right]$ and $W_B := \sum_{i=1}^m W_{i,i} = W_{1,1} + W_{2,2} + \cdots +W_{m,m}$. 
\end{remark}

\begin{lemma}
\label{le:ConvexityOfTr(W^2)}
If $W$ is an Hermitian matrix, then the matrix function $\tr(W^2)$ is a strictly convex function.
\end{lemma}

We present the proof of above lemma in Appendix \ref{app:leConvexityOfTr(W^2)}. Next, the following fact is known in algebraic geometry.

\begin{lemma}
\label{lem:Existence}
    (i) If the subspace $V \subseteq \mathbb{C}^m \otimes \mathbb{C}^n$ with $\dim V > (m-1)(n-1)$, then $V$ contains at least one product vector.

    (ii) If $W \in \mathcal{BP}_{m,n}$, then $W$ has no more than $(m-1)(n-1)$ negative eigenvalues.

(iii) Suppose the subset $X\subset \bbC^m\otimes\bbC^n$ consists of bipartite vectors of Schmidt rank at most $k$. Then the subset is an affine variety of dimension $k(m+n-k)-1$ in projectivization.
    \qed
\end{lemma}

We summarize the material related to Lemma \ref{lem:Existence} and present it as a remark in Appendix \ref{app:lemExistence}.


\subsection{quantum information: bipartite systems} 
\label{sec:qi=bipartite}

\begin{definition}
\label{def:ew_def2}
An Hermitian matrix $H \in \mathbb{M}_m \otimes \mathbb{M}_n = \mathbb{M}_{mn}$ is called an entanglement witness (EW) with the normalization, i.e., $\tr H = 1$, if  

    (i) $H = H^*$; 

    (ii) $H$ is not positive semi-definite, i.e., there exists $\vec{x} \in \mathbb{C}^m \otimes \mathbb{C}^n$ such that $\vec{x}^* H \vec{x} < 0$; 

    (iii) $(\vec{a}\otimes \vec{b})^* H (\vec{a}\otimes \vec{b}) \geq 0, \quad \forall \vec{a}\otimes \vec{b}\in \mathbb{C}^m \otimes \mathbb{C}^n$. 
    \qed
\end{definition}
From the definition, we present the following fact.

\begin{lemma}
\label{lem:entangled_eq}
    A bipartite state $\r$ is entangled if and only if there exists an entanglement witness $W$ such that $\tr(W\rho) < 0$.
    \qed
\end{lemma}

It is known that determining whether a bipartite positive-partial-transpose (PPT) state is entangled is an NP-hard problem. The following notion of DEWs has been proposed for effectively characterizing non-PPT states.

\begin{definition}
\label{def:DEW2}
A bipartite entanglement witness $H \in \mathbb{M}_m \otimes \mathbb{M}_n$ is a decomposable EW (DEW), if there exist positive semi-definite Hermitian matrices $X, Y \in \mathbb{M}_m \otimes \mathbb{M}_n$ such that      \begin{eqnarray}
\label{eq:H=X+YG}
    H = X + Y^{\Gamma},
\end{eqnarray}
where $Y^{\Gamma}$ denotes the partial transpose of $Y$ with respect to the first subsystem. Otherwise, $H$ is called a nondecomposable EW (NDEW).
\end{definition}
\begin{lemma}
\label{lem:dew_eq}
An EW $H \in \mathbb{M}_m \otimes \mathbb{M}_n$ is a DEW if and only if $\tr(H\sigma) \ge 0, \quad \forall \text{ PPT state } \sigma $.
\end{lemma}

We present a novel proof of Lemma \ref{lem:dew_eq} in Appendix \ref{app:lemdew_eq}, accompanied by an illustrative example that constructs a new EW by using a known EW. 

\begin{remark}
\label{rem:abWab}
Note that $\langle a, b | W | a, b \rangle = (|a \rangle \otimes |b \rangle)^*W(|a \rangle \otimes |b \rangle) = \langle b| \left[ (\langle a| \otimes I_n) W (|a \rangle \otimes I_n) \right]|b \rangle$. So the condition in the preceding definition is equivalent to $(\langle a| \otimes I_n) W (|a \rangle \otimes I_n)$ being a positive semi-definite matrix, since $|b \rangle$ is arbitrary. 
\qed
\end{remark}

Next, we introduce some facts on bipartite EWs from Lemma 8 of \cite{2025Spectral}.
\begin{lemma}
	\label{le:EW=general}
	Let $W$ be a normalized $m\times n$ EW. Then
	
    (i) $\tr(W^2)\in (\frac{1}{mn-1},1]$.

	(ii) $\l_{mn}\in [-\frac{1}{2},0)$. If $\l_{mn}=-\frac{1}{2}$, then W is optimal. 
	
	(iii)  $\l_1\in (\frac{1}{mn-1},1)$.

	(iv) For $m=2$: (a) $\l_2+\l_{2n}\ge 0$, (b) $\sum_{i=3}^{2n} \l_i\ge -\frac{1}{2+2\sqrt{2}}$ and (c)
	$\sum_{i=k}^{2n} \l_{i}\ge -\frac{1}{2}$ for any $4\le k\le 2n$.
	
	(v) For $m=n$, $\cN(W)\le \frac{m-1}{2}$. If the inequality is saturated, then $W$ is optimal and has eigenvalues $\frac{1}{m}$ with multiplicity $\frac{m(m+1)}{2}$, and $-\frac{1}{m}$ with multiplicity $\frac{m(m-1)}{2}$. 
\qed
	\end{lemma}

The following fact extends Lemma \ref{le:EW=general} for DEWs.

\begin{lemma}
	\label{le:lambdaIN[-1/2,1)}

Let $W$ represent an element of $m\times n$ normalized DEWs ($m\le n$). 

(i) For each $x\in (\frac{1}{mn-1},1]$, there exists a $W$ such that $\tr(W^2)=x$. The infimum of $\tr(W^2)$ is $\frac{1}{mn-1}$ and not attainable. The supremum of $\tr(W^2)$ is 1 and attainable.

(ii) $\tr(W^2)=1$ if and only if $W$ is the partial transpose of a pure entangled state.

(iii) For each $x\in [-\frac{1}{2},0)$, there exists a $W$ such that $\l_{mn}(W)=x$. The infimum of $\l_{mn}$ is $-\frac{1}{2}$ and 
attainable. The supremum of $\l_{mn}$ is $0$ and not
attainable.

(iv)
$\l_{mn}(W)=-\frac{1}{2}$ if and only if $W$ is the partial transpose of a pure state whose two nonzero Schmidt coefficients are both $\frac{\sqrt{2}}{2}$. 
	
(v) For each $x\in (\frac{1}{mn-1},1)$, there exists a $W$ such that $\l_1(W)=x$. The infimum of $\l_1(W)$ is $\frac{1}{mn-1}$ and not attainable. The supremum of $\l_1(W)$ is 1 and not attainable.

(vi) $0<\cN(W)\le \frac{m-1}{2}$. For each $x\in (0,\frac{m-1}{2}]$, there exists a $W$ such that $\cN(W)=x$. The supremum of $\cN(W)$ is $\frac{m-1}{2}$ and attainable. The infimum of $\cN(W)$ is 0 and not attainable.

(vii) For the integer $1\le j\le \lfloor\frac{n}{m}\rfloor$, we define the maximally entangled state $\ket{\Phi_j}:=\frac{1}{\sqrt{m}}\sum_{i=1}^m \ket{i,i+(j-1)m}\in \bbC^m\otimes \bbC^n$. Then $\cN(W)=\frac{m-1}{2}$ if and only if under local unitary equivalence, $W$ can be written as the convex combination of $\proj{\Phi_1}^\G,\cdots,\proj{\Phi_k}^\G$, where $1\le k\le \lfloor \frac{n}{m} \rfloor$.

(viii) Let $m=2$. The infimum of $\sum_{i=3}^{2n} \l_i$ is $-\frac{1}{2+2\sqrt{2}}$, which is attained by $(x_1\ket{11}+x_2\ket{22})
(x_1\bra{11}+x_2\bra{22})^\G$, where $x_1^2+x_1x_2=1-\frac{1}{2+2\sqrt{2}}$ (numerically, $x_1=0.92388, x_2=0.382683$). For any $4\le k\le 2n$, the infimum of 
$\sum_{i=k}^{2n} \l_{i}$ is $-\frac{1}{2}$, which can be attained by the partial transpose of a qubit-qudit maximally entangled state.

(ix) For $3\le m\le n$, the infima of $\l_{mn}+\l_{mn-1}$ and  $\l_{mn}+\l_{mn-1}+\l_{mn-2}$
are $-\frac{\sqrt{2}}{2}$ and $-1$ respectively. The first infimum can be attained by the partial transpose of a pure state whose nonzero Schmidt coefficients are $\frac{\sqrt{2}}{2},\frac{1}{2},\frac{1}{2}$.  The second can be attained by the partial transpose of a two-qutrit maximally entangled state. 
\qed
\end{lemma}

The following fact extends Lemma \ref{le:EW=general} for NDEWs.
\begin{lemma}
	\label{le:NDEW}
	Let $W$ represent an element of $m\times n$ normalized  NDEWs ($m\le n$). Then
	
	(i) The supremum of $\tr(W^2)$ is 1.
	
	(ii) The infimum of $\tr(W^2)$ is not attainable.
	
    (iii) The supremum of $\l_{mn}$ is 0 and  not attainable.

    (iv) The infimum of $\l_{mn}$ is $-\frac{1}{2}$ and not attainable.

   (v) The supremum of $\l_1$ is 1 and not attainable.

	(vi) The infimum of $\l_1$ is not attainable.
	
	(vii) For $m=2$ and $m=n\ge 3$, the supremum of $\cN(W)$ is $\frac{m-1}{2}$.  
	
	(viii) For $m=n=3$, the supremum of $\cN(W)$ is not attainable.
	
	(ix) The infimum of $\cN(W)$ is $0$ and not attainable.
\qed
	\end{lemma}

\begin{lemma}
\label{le:W=xI_mn-|ps><ps|}
(i) Let $W=x I_{mn}-\proj{\ps}$ be a bipartite EW, where $x\in(0,1)$ is lower bounded by the maximum Schmidt coefficient of entangled state $\ket{\ps}$. Then $W$ is a bipartite DEW.    

(ii) Let $\ket{\ps}$ be an $n$-partite genuinely entangled pure state in $\cH_D$. Let $W=x I_D-\proj{\ps}$ be an $n$-partite EW, where $x\in(0,1)$ is lower bounded by the maximum Schmidt coefficient of $\ket{\ps}$ w.r.t. any system bipartition. Then $W$ is an $n$-partite DEW.    
\end{lemma}
\begin{proof}
(i) One can show the claim by straightforward calculation and definition of DEW, namely the partial transpose of $W$ is semidefinite positive.    

(ii) The claim follows from (i).
\end{proof}

The following fact is from Lemma 9 in \cite{2025Spectral}.
\begin{lemma}
\label{le:existence_0f_10000}
    Any deficient-rank NPT (resp. PPT) entangled state can be detected by a normalized DEW (resp. NDEW) which is arbitrarily close to a pure state. Consequently, normalized $W$ exists for both DEWs and NDEWs such that $\lambda(W)$ is arbitrarily close to $(1,0,\cdots,0)$.
\end{lemma}

\subsection{Quantum information: multipartite system}
\label{sec:qi=multipartite}

\begin{definition}
\label{def:DEWn}
An $n$-partite entanglement witness $H \in \bbM_D:= \mathbb{M}_{d_1} \otimes... \otimes \mathbb{M}_{d_n}=\cB(\cH_{A_1}\otimes...\otimes\cH_{A_n})$ is a decomposable EW (DEW), if there exist positive semi-definite Hermitian matrices $X_1, X_2,... \in \bbM_D$ such that      
\begin{eqnarray}
\label{eq:H=X+YG}
        H = X_1 + \sum_{j>1} X_j^{\Gamma_{S_j}},
\end{eqnarray}
where $\Gamma_{S_j}$ denotes the partial transpose with respect to the subsystem family $S_j\subset\{A_1,...,A_n\}$, such that $\abs{S_j}\le \lfloor{n \over 2}\rfloor$ and $S_j\cup S_k\ne \{A_1,...,A_n\}$. Otherwise, $H$ is called a nondecomposable EW (NDEW). 
    \qed
\end{definition}

The following fact is from \cite{2025Sufficient}.
\begin{lemma}
\label{le:multi-AS}
Let \( \rho \) be a normalized \( N \)-qudit state acting on a Hilbert space of total dimension \( D = d^N \). If  
\[
\operatorname{Tr} \left( \rho^2 \right) \leq \frac{1}{D - A}, \tag{15}
\]
with \( A = 2^{2-N} \), then \( \rho \) is absolutely fully separable.  

For the specific case of qubits, \( d = 2 \), this bound can be improved to \( A = \frac{\beta \cdot 2^N}{\beta + 3^N} \), where \( \beta = 54/17 \).
\qed
\end{lemma}

The following fact is from \cite{2003LGH}. 
\begin{lemma}
\label{le:m-partite=sep}
If an \( m \)-partite state \( \rho: H_1 \otimes \cdots \otimes H_m \to H_1 \otimes \cdots \otimes H_m \) satisfies \( \| \rho - I/d \|_2 \leq 1/2^{m/2-1} d \), where \( d = \dim(H_1 \otimes \cdots \otimes H_m) \), then it is separable.

It actually gives the (negligibly) tighter statement with \( 2^{m/2-1} \sqrt{d(d-2^{-(m-2)})} \) in the denominator.
\qed    
\end{lemma}

The definition of bipartite optimized EWs has been proposed in \cite{lewenstein2000optimization1}. We extend it to multipartite scenarios. 
\begin{definition}
\label{def:finer}
    For an entanglement witness (EW), we define the set
    \begin{eqnarray}
        D_W := \{ \rho \in B(\mathbb{C}^{d_1} \otimes \cdots \otimes \mathbb{C}^{d_n}) \mid \operatorname{Tr}(W \rho) < 0 \}.
    \end{eqnarray}
    We say that $W_1$ is finer than $W_2$, provided $D_{W_1} \supseteq D_{W_2}$. Furthermore, an EW $W$ is optimal when no EW is finer than $W$.
\end{definition}
\begin{theorem}
\label{lem:finer_P}
    Let $W_1$ and $W_2$ be two normalized $d_1 \times d_2 \times \cdots \times d_n$ EWs. If $W_2$ is finer than $W_1$, then there is any positive operator $P$ and an $k \in [0, 1)$ such that 
    \begin{eqnarray}
    \label{eq:W1W2andP}
        W_1=(1-k)W_2+k P.
    \end{eqnarray}
\end{theorem}

We provide an elegant proof of Theorem \ref{lem:finer_P} by constructing a special infimum. These details are given in \ref{app:lemfiner_P}.

\begin{lemma}
    Let W be an arbitrary normalized $d_1 \times d_2 \times \cdots \times d_n$ EW. Then $W$ is optimal if and only if for any positive operator $P$, the operator $W-P$ is no longer block-positive.
\end{lemma}
\begin{proof}
First, we prove the "only if" part. Since $P \geq 0$, we have 
\begin{eqnarray}
    \tr[(W-P)\r] = \tr (W\r) - \tr (P\r) \le \tr (W\r), 
\end{eqnarray}
which implies $D_{W-P} \supseteq D_{W}$. Then we prove by contradiction. Suppose There is a positive operator $P$ such that $W-P$ is block-positive. Then since there must exist $\r_1$ with $\tr (W\r_1)<0$, we also have $\tr [(W-P)\r_1]<0$. This is a contradiction. The "if" part is clear using Theorem \ref{lem:finer_P}. So the proof is complete. 
\end{proof}

We need check whether some of facts in Lemma \ref{le:lambdaIN[-1/2,1)} and \ref{le:NDEW} can be extended to multipartite scenarios. This is the main motivation of next section. 
For this purpose, we present some facts on multipartite EWs of the same systems. 

\begin{lemma}
\label{le:n-partite EW split=EW}
Let $W\in \cM_{d_1}\otimes...\otimes \cM_{d_n}\cong\cM_{d_1...d_n}=\cM_D=\cB(\cH_D)$ be an $n$-partite EW of system $A_1,...,A_n$. For each $j$, we split $A_j$ into $m_j\ge1$ disjoint subsystems $A_{j,1},...,A_{j,m_j}$ and rename $W$ as $W'$. Then 

(i) $W'$ is an $(\sum^n_{j=1} m_j)$-partite EW of systems $A_{1,1},...,A_{n,m_n}$. 

(ii) If $W$ is a DEW then $W'$ is also a DEW. 
\end{lemma}

We present the proof of Lemma \ref{le:n-partite EW split=EW} in \ref{app:len-partite EW split=EW}.

The converse of Lemma \ref{le:n-partite EW split=EW} (i) fails. That is, the concentration of systems of an EW may result in a non-EW Hermitian matrix. For example we consider the three-qubit EW $W_1={1\over3}I_8-\proj{W}$ where $\ket{W}={1\over\sqrt3}(\ket{001}+\ket{010}+\ket{100})$ of system $A,B$ and $C$. If we merge system $A,B$ then $W_1$ is no longer a $2\times4$ EW, because the maximum overlap between $\ket{W}$ and bipartite product states in $\bbC^2\otimes\bbC^4$ equals $\sqrt{2/3}$. 
Indeed, the maximum squared overlap of $W$ with tripartite product states such as $|000\rangle$ and $ |001\rangle$ is $\frac{1}{3}$. Therefore, for any tripartite product state $|\phi \rangle$, we have
\begin{eqnarray}
    \langle \phi | W_1 | \phi \rangle = \frac{1}{3} \langle \phi | \phi \rangle - \langle \phi | W \rangle^2 \geq \frac{1}{3} - \frac{1}{3} = 0.
\end{eqnarray}
At the same time, 
\begin{eqnarray}
    \langle W | W_1 | W \rangle= \frac{1}{3} \langle W | W \rangle - \langle W | W \rangle^2 = \frac{1}{3} - 1 = -\frac{2}{3} <0.
\end{eqnarray}
So $W_1$ is a three-qubit EW. On the other hand, $W$ can be written as
\begin{eqnarray}
    |W\rangle = \frac{1}{\sqrt{3}} \left[|00\rangle_{AB}|1\rangle_C + (|01\rangle_{AB} + |10\rangle_{AB})|0\rangle_C\right].
\end{eqnarray}
Let
\begin{eqnarray}
     |\psi \rangle = \frac{|01\rangle_{AB} + |10\rangle_{AB}}{\sqrt{2}} \otimes |0\rangle_C.
\end{eqnarray}
Under AB-C bipartition, this is a product state. We have
\begin{eqnarray}
     \langle \psi | W_1 | \psi \rangle = \frac{1}{3} - \frac{2}{3} = -\frac{1}{3} < 0 ,
\end{eqnarray}
which contradicts with the definition of EWs.

Next if $W'$ defined in the way of Lemma \ref{le:n-partite EW split=EW} is an NDEW then $W$ may not be an NDEW. It relies on whether a PPT entangled state remains PPT entangled under system split. The latter is not true generally, say the rank-four 3-qubit PPTES $\r$ constructed by UPB. Actually, if $\r$ is detected by a three-qubit NDEW, which is then not a $2\times4$ EW when $\r$ becomes a bipartite separable state.

As an example of $\r$, we consider a three-qubit system, and define the following four orthogonal product states
\begin{align}
    &|\psi_0\rangle = |0\rangle_A |0\rangle_B |0\rangle_C,\\
    &|\psi_1\rangle = |1\rangle_A |-\rangle_B |+\rangle_C,\\
    &|\psi_2\rangle = |+\rangle_A |1\rangle_B |-\rangle_C,\\
    &|\psi_3\rangle = |-\rangle_A |+\rangle_B |1\rangle_C,
\end{align}
which constitute the UPB. Here $|+\rangle = \frac{|0\rangle+|1\rangle}{\sqrt{2}}$ and $|-\rangle = \frac{|0\rangle-|1\rangle}{\sqrt{2}}$. Thus, we have
\begin{eqnarray}
    \rho_{UPB} = \frac{1}{4} \left( I - \sum_{i=0}^3 |\psi_i\rangle\langle\psi_i| \right)
\end{eqnarray}
where I denotes the $8 \times 8$ identity matrix. Since the total dimension of the Hilbert space is eight, and we have subtracted four dimensions, the rank of $\rho_{UPB}$ is exactly four. This corresponds to the rank-four 3-qubit PPTES mentioned before.

\begin{lemma}
\label{le:GHZ=DEW}
Let $2\le d_1\le ...\le d_n$ and $\ket{\ghz_{d_1}}={1\over\sqrt d_1}\sum^{d_1-1}_{j=0}\ket{j}^{\otimes n}$. Let $W=x I_D-\proj{\ghz_{d_1}}$ be an EW on the $n$-partite Hilbert space $\bbC^D\cong \bbC^{d_1}\otimes...\otimes\bbC^{d_n}$ of system $A_1,...,A_n$, where $x\in[1/d_1,1)$. Then $W$ is a $k$-partite DEW of systems $B_1,...,B_k$ whose union is $A_1,...,A_n$, for any $k\in[1,n]$.   
\end{lemma}
\begin{proof}
The claim follows from Lemma \ref{le:W=xI_mn-|ps><ps|} (ii).   
\end{proof}

\begin{lemma}
\label{le:fullClassiEigen}
    Suppose $|v\rangle \in \mathcal{H} = \mathbb{C}^{d_1} \otimes \mathbb{C}^{d_2} \otimes \cdots \otimes \mathbb{C}^{d_n}$ and assume the tensor rank of $|v\rangle$ is r. Here, any  $|v\rangle$ can be expressed in the form $|v\rangle=\sum_{i=1}^{r}\alpha_i|a_{1,i}\rangle\otimes|a_{2,i}\rangle\otimes \cdots \otimes |a_{n,i}\rangle$ where $|a_{1,i}\rangle, \ |a_{2,i}\rangle, \ \cdots , \ |a_{n,i}\rangle$ form an orthonormal basis of three subsystems of $\mathcal{H}$, respectively, namely as follows
    \begin{eqnarray}
        \langle a_{1,i} | a_{1,j} \rangle =\cdots =\langle a_{n,i} | a_{n,j} \rangle=\delta_{ij},
    \end{eqnarray}
    and $\alpha_1\geq\alpha_2\geq\cdots\geq\alpha_r>0, \ i=1, \cdots r$,  $\sum_{i=1}^r\alpha_i^2=1$. If $|v\rangle$ satisfies the conditions as before, i.e., $|v\rangle$ is fully classical, then the matrix $(|v\rangle\langle v|)^{\G}$ has nonzero eigenvalues
    \begin{eqnarray}
        \alpha_k^2, \quad k=1, \cdots , r
    \end{eqnarray}
    and
    \begin{eqnarray}
        \pm \alpha_k \alpha_l, \quad k,l=1, \cdots , r;k \neq l
    \end{eqnarray}
    where $\G$ represents the partial transpose of the matrix with respect to any one of the three parts.
\end{lemma}

Details of the proof of Lemma \ref{le:fullClassiEigen} are given in \ref{app:lefull}. Furthermore, we obtain that the number of eigenvalues $\alpha_k^2$ is r, and the number of eigenvalues $\pm \alpha_k \alpha_l$ is $r^2 - r$.

\begin{lemma}
\label{le:2^3}
    The negative subspace of a 3-qubit EW $W \in \mathcal{H} = \mathbb{C}^2 \otimes \mathbb{C}^2 \otimes \mathbb{C}^2$ has dimension at most four.
\end{lemma}

We present the proof of Lemma \ref{le:2^3} in \ref{app:le2^3} using the well-known projection method to solve this problem, and then applying the dimension formula for intersections and unions.

Next we generalize the lemma to higher dimensions, thereby obtaining Lemma \ref{le:2^n} and Corollary \ref{coro:d1...dn}. The corresponding proof is placed in \ref{app:le2^n}.

\begin{lemma}
\label{le:2^n}
    The negative subspace of an $n$-qubit EW $W \in \mathcal{H} = \mathbb{C}^2 \otimes \mathbb{C}^2 \otimes \cdots \otimes \mathbb{C}^2$ has dimension at most $2^n-n-1$.
\end{lemma}

In particular, when $n=2$ and $n=3$, Lemma \ref{le:2^n} reduces to the result in \cite{2025Spectral} and Lemma \ref{le:2^3}, respectively. More generally, the negative subspace of a n-qubit EW $W \in \mathcal{H} = \mathbb{C}^{d_1} \otimes \mathbb{C}^{d_2} \otimes \cdots \otimes \mathbb{C}^{d_n}$ has dimension at most 
\begin{eqnarray}
    n+\prod_{i=1}^nd_i+\sum_{i=1}^nd_i-1.
\end{eqnarray}

\begin{corollary}
\label{coro:d1...dn}
    For a n-qubit EW $W \in \mathcal{H} = \mathbb{C}^{d_1} \otimes \mathbb{C}^{d_2} \otimes \cdots \otimes \mathbb{C}^{d_n}$, we have
    \begin{eqnarray}
        \l_{1-n+\sum_{i=1}^nd_i} \geq 0.
    \end{eqnarray}
    In particular, when $d_1=d_2= \cdots =d_n=2$, we have $\l_{n+1} \geq 0$. 
\end{corollary}

\section{Results on eigenvalues of $2 \times n$ EWs}
\label{sec:resBi}
In this section, we conduct an in-depth investigation and extension of the previous results of bipartite EWs. Firstly, we examine properties of a single eigenvalue. Using Lemma \ref{le:EW=general} (iv), we have the following fact (i). Also, applying Lemma \ref{le:EW=general} (iv), we can obtain that for any normalized $2\times n$ EW, 
\begin{align}
    3\l_3 + 4\sum_{i=4}^{2n} \l_{i} =&3\sum_{i=3}^{2n} \l_{i} +\sum_{i=4}^{2n} \l_{i}\\
    \geq&3\cdot\left( -\frac{1}{2+2\sqrt{2}} \right) + \left( -\frac{1}{2} \right)=-1.121320 \cdots.
\end{align}
However, we have the following stronger conclusion (ii).
\begin{theorem}
\label{le:2tin1234}
    Let $W$ be a normalized $2\times n$ EW. 

    (i) For a single eigenvalue range, we have the following conclusion.
    \begin{eqnarray}
        \l_{3} \geq -\frac{1}{(2+2\sqrt{2})(2n-2)}, \quad \l_4 \geq -\frac{1}{2(2n-3)}.
    \end{eqnarray}
    Furthermore, for any $4\le k\le 2n$, we have
    \begin{eqnarray}
        \l_k \geq -\frac{1}{2(2n-k+1)}. 
    \end{eqnarray}

    (ii) Let $W$ be a normalized $2\times n$ EW. Then 
    \begin{eqnarray}
        3\l_3 + 4\sum_{i=4}^{2n} \l_{i} \geq -1.
    \end{eqnarray}

    (iii) Let $W$ be a normalized $2\times n$ EW. Then 
    \begin{eqnarray}
        2p\sum_{i=4}^{2n} \l_{i} + (p^2-1)\l_3 \geq -1, 
    \end{eqnarray}
    where $1 \le p \le 1+\sqrt{2}$.

    (iv) Let $W$ be a normalized $2\times n$ EW. Then, 
    \begin{eqnarray}
        (m^2+2m-1)\sum_{i=4}^{2n} \l_{i} + (m^2-1)(\l_3 + \l_2) \geq -1,
    \end{eqnarray}
    where $m \geq 1$.
\end{theorem}

The proof of the above theorem is given in the beginning of  \ref{app:2tinEW}. 


\begin{lemma}
\label{le:2n13>=}
    Let $W$ be a normalized $2\times n$ EW. Then
    \begin{eqnarray}
    \label{ineq:2n13>=}
        \l_{2n} + \sqrt{\l_1 \l_3} \geq 0.
    \end{eqnarray}
\end{lemma}
\begin{proof}
For the two-qubit case, we have $\l_4 + \sqrt{\l_1 \l_3} \geq 0$ from Theorem 3 of  \cite{Johnston_2018}. For $2 \times n $ ($n>2$) systems, we assume that there exists an EW $W$ such that 
\begin{eqnarray}
\label{ineq:2n13<}
    \l_{2n} + \sqrt{\l_1 \l_3} < 0. 
\end{eqnarray}
Up to local unitary equivalence, we can assume the eigenvector corresponding to $\l_{2n}$ is $\cos \theta {|00\rangle} + \sin \theta {|11\rangle}$ with $\theta \in (0, \frac{\pi}{2})$. Next we project $W$ onto the subspace spanned by $\{|0\rangle, |1 \rangle\} \otimes \{|0\rangle, |1 \rangle\}$, and then we obtain a $W' \in \mathcal{BP}_{2,2}$. Let eigenvalues of $W'$ be $\l'_i$, $i=1,2,3,4$, then one can see $\l'_4 = \l_{2n}<0$. Thus, W' is a two-qubit EW. Furthermore, using Lemma 1 of \cite{2025Spectral}, we have $\l'_1 \le \l_1 $ and $\l'_3 \le \l_3 $. Using Inequality \eqref{ineq:2n13<}, we obtain that
\begin{eqnarray}
    \l'_4 + \sqrt{\l'_1 \l'_3} \le \l_{2n} + \sqrt{\l_1 \l_3} <0, 
\end{eqnarray}
which leads to a contradiction. Therefore, the proof is complete. 
\end{proof}

The discussion in Remark \ref{re:constrOfW_SS} motivates the following construction.

\begin{example}
\label{eg:genuine=}
Let
\begin{eqnarray}
    |\psi_i \rangle = \frac{|0,2i\rangle+|1,2i+1\rangle}{\sqrt{2}}
\end{eqnarray}
and 
\begin{eqnarray}
    W=\frac1k\sum_{i=0}^{k-1} |\psi_i\rangle\langle\psi_i|^\Gamma, 
\end{eqnarray}
then $W$ is an $2 \times n$ EW where $n=2k$. One can see that $W$ has $3k$ eigenvalues equal to $\frac{1}{2k}$ and $k$ eigenvalues equal to $-\frac{1}{2k}$. Thus it satisfies 
\begin{eqnarray}
    \lambda_{2n}+\sqrt{\lambda_1\lambda_3}=\lambda_{4k}+\sqrt{\lambda_1\lambda_3}=0. 
\end{eqnarray}
\end{example}

\section{Results on multipartite EWs}
\label{sec:res}

We first present some research findings on multipartite EWs. 
\begin{theorem}
\label{le:multiEWTRW^2}
    Let $W$ be a normalized $d_1 \times d_2 \times \cdots \times d_n$ EW. Then the lower bound of $\tr (W^2)$ is $\frac{1}{d_1d_2 \cdots d_n-1}$ and not attainable.
\end{theorem}
\begin{proof}
We define $D=d_1d_2 \cdots d_n$. Using  the Cauchy-Schwarz inequality, we have
\begin{align}
    \tr (W^2) &=\l_1^2+\cdots + \l_D^2 = (\l_1^2+\cdots + \l_D^2) + \l_{d_1d_2 \cdots d_n}^2\\
    &= \frac{(\l_1^2+\cdots + \l_{D-1}^2)(1^2 + \cdots +1^2)}{D-1} + \l_D^2\\
    &\geq\frac{(\l_1+\cdots + \l_{D-1})^2}{D-1} + \l_{D}^2\\
    &=\frac{(1-\l_D)^2}{D-1} + \l_D^2 > \frac{1}{D-1},
\end{align}
where the last inequality is due to  $\l_D<0$ by definition of EWs. 
\end{proof}

The value $\tr W^2=1$ is achievable when we take the EW being the partial transpose of a GHZ state. 
Here is the detail. Let $d_1 = \min_{1\le k \le n} \{d_k\} $, then we define
\begin{eqnarray}
    \Psi_{d_1}=\frac{1}{\sqrt{d_1}}\sum_{i=1}^{d_1}|ii\cdots i\rangle \in \mathbb{C}^{d_1} \otimes \mathbb{C}^{d_2} \otimes \cdots \otimes \mathbb{C}^{d_n},  
\end{eqnarray}
and we claim that the upper bound one is achievable when we take the EW being $|\Psi_{d_1} \rangle \langle\Psi_{d_1} |^ \G$ where $\G$ acts on any partial systems. Indeed, the eigenvalues of $|\Psi_{d_1} \rangle \langle\Psi_{d_1} |^ \G$ are precisely $\frac{1}{d_1}$ and $-\frac{1}{d_1}$, with algebraic multiplicities $\frac{d_1^2+d_1}{2}$ and $\frac{d_1^2-d_1}{2}$, respectively. For instance, when $d_1=3$, we directly construct the simplest possible matrices of order $27 \times 27$ in \ref{app:ExplicitConstructionGHZ}. We then obtain all eigenvalues of $|\Psi_3 \rangle \langle\Psi_3 |^ \G$, i.e., $\frac{1}{3}$ and $-\frac{1}{3}$, with algebraic multiplicities six and three, respectively.

On the other hand, the upper bound of $\tr W^2$ is hard to derive at present, because it is unknown whether the maximum eigenvalue of a three-qubit normalized NDEW may exceed one. This is also why Lemma \ref{le:multi-DEW=lambda1<1} does not work for NDEW yet. 
Further, the infimum of $\tr(W^2)$ is unknown for both bipartite and multipartite NDEW $W$ by Lemma \ref{le:NDEW}.

\subsection{Multipartite DEWs}

In this subsection, we present the main result of this paper. We investigate the supremum and infimum of maximum and minimum eigenvalues of multipartite DEWs. 
The following result partially extends Lemma \ref{le:lambdaIN[-1/2,1)} (i) and \ref{le:NDEW} (ii) on the lower bound and infimum of $\tr (W^2)$.

Before arriving at this result, we first derive a lemma for the special case, which serves as a preparation for the later construction.

\begin{lemma}
\label{le:GHZ_n}
    Let
    \begin{eqnarray}
        |GHZ_{n}\rangle = \frac{1}{\sqrt{2}} (|0\rangle^{\otimes n} + |1\rangle^{\otimes n}), 
    \end{eqnarray}
    and define 
    \begin{eqnarray}
        W = \frac{1}{k} \sum_{S_1, \cdots, S_k \in {A_1, \cdots, A_n}} |GHZ_{n}\rangle \langle GHZ_{n}|^{T_{S_i}}. 
    \end{eqnarray}
    Then $W$ is a DEW if and only if there is a negative eigenvalue of $W$. 
\end{lemma}

\begin{proof}
Firstly, we claim that $W$ is decomposable. Indeed, for $W$, let $X_1$ in Definition \ref{def:DEWn} be $O$, then one can see $W$ is decomposable. 

Secondly, for any separable state $|\psi \rangle = |i_1\rangle \otimes |i_2\rangle \otimes \cdots \otimes |i_n\rangle$, the following equation holds
\begin{align}
    \langle \psi |W| \psi \rangle =& \tr (|\psi \rangle \langle \psi |W)\\
    =& \tr \left[ \frac{1}{k} \sum_{S_1, \cdots, S_k \in {A_1, \cdots, A_n}} (|\psi \rangle \langle \psi |) (|GHZ_{n}\rangle \langle GHZ_{n}|)^{T_{S_i}} \right]\\
    =&\tr \left[ \frac{1}{k} \sum_{S_1, \cdots, S_k \in {A_1, \cdots, A_n}} (|\psi \rangle \langle \psi |)^{T_{S_i}} (|GHZ_{n}\rangle \langle GHZ_{n}|) \right]. 
\end{align}
Let the partial transpose of the state $|\psi \rangle$ be $|\psi \rangle ^*$, then
\begin{eqnarray}
    \langle \psi |W| \psi \rangle =& \tr \left[ \frac{1}{k} \sum_{S_1, \cdots, S_k \in {A_1, \cdots, A_n}} (|\psi \rangle ^* \langle \psi | ^*) (|GHZ_{n}\rangle \langle GHZ_{n}|) \right]
    \geq 0. 
\end{eqnarray}
Therefore, using the definition of EW, we conclude that $W$ is a DEW if and only if there is a negative eigenvalue of $W$. 
\end{proof}

\begin{theorem}
\label{le:d1-dnTrW^2}
Let $W$ be a normalized $d_1 \times d_2 \times \cdots \times d_n$ EW. 
    
(i) If $W$ is a DEW then the infimum of $\tr (W^2)$ is also $\frac{1}{d_1d_2 \cdots d_n-1}$. 

(ii) For each $x\in (\frac{1}{d_1d_2...d_n-1},1]$, there exists a DEW $W$ such that $\tr(W^2)=x$. 

(iii) If $W$ is a DEW then the supremum of $\tr (W^2)$ is one and attainable.
\end{theorem}
\begin{proof}
(i) It suffices to prove that there exists a DEW whose limit of the trace is exactly $\frac{1}{D-1}$. We consider 
\begin{eqnarray}
\label{eq:W1=xI_D-|psi><psi|}
    W_1 = x I_D-|\psi\rangle \langle \psi|
\end{eqnarray}
where $I_D$ is an order-$D$ identity matrix, and $\ket{\ps}$ is an $n$-partite genuinely entangled pure state. Then we claim that $\widetilde{W_1}:=\frac{W_1}{xD-1}$ is an EW if $1> x\geq x_0$ where 
\begin{eqnarray}
    x_0 := \sup_{||a_j|| = 1} |\bra{a_1,a_2} |\psi\rangle|^2,
\end{eqnarray}
where the union of systems of pure states $\ket{a_1}$ and $\ket{a_2}$ is the system of $\ket{\ps}$. Indeed, for an arbitrary product state $|\phi_{\text{sep}} \rangle:=\ket{a_1,a_2}$, that can be expressed as a ${D}$-dimensional vector with $||a_j|| = 1$, we have 
\begin{eqnarray}
    \langle \phi_{\text{sep}}|W_1|\phi_{\text{sep}} \rangle = x\langle \phi_{\text{sep}}|\phi_{\text{sep}} \rangle - \langle \phi_{\text{sep}}|\psi \rangle^2 \geq \langle \phi_{\text{sep}}|\psi \rangle^2 - \langle \phi_{\text{sep}}|\psi \rangle^2 =0,
\end{eqnarray}
which implies that $\widetilde{W_1}$ is a DEW by Lemma \ref{le:GHZ=DEW}.
Thus for $x\in [x_0, 1)$, 
\begin{align}
    \tr (\widetilde{W_1}^2) &= \frac{\tr(x^2I_D-2x|\psi\rangle \langle \psi|+|\psi\rangle \langle \psi|\psi\rangle \langle \psi|)}{(x{D}-1)^2}\\
    &=\frac{x^2D-2x+1}{(xD-1)^2}. 
\end{align}
Therefore, 
\begin{eqnarray}
    \lim_{x\rightarrow1^-} \tr (\widetilde{W_1}^2) = \frac{1}{D-1}.
\end{eqnarray}
So the proof is complete. 

(ii) The claim follows by replacing a normalized DEW $W_1$ by $W=(W_1+y\proj{\ps})/(1+y)$ where $y>0$ and $\ket{\ps}$ is the eigenvector of the maximum eigenvalue $W_1$. One can make $\tr W^2=x$ by a suitable $y$.

(iii) We need show that $\tr W^2\le1$, where $W$ is a multipartite DEW. We first consider a normalized DEW $W=P+Q^{\G_A}+R^{\G_B}$ where $\frac{P+Q^{\G_A}}{\tr(P+Q^{\G_A})}$ is also a normalized DEW. Then $\tr(W^2) \le 1$. The detailed proof is provided in \ref{app:UseConvexity}.

Next, we consider an $n$-partite entanglement witness $H \in \bbM_D:= \mathbb{M}_{d_1} \otimes... \otimes \mathbb{M}_{d_n}=\cB(\cH_{A_1}\otimes...\otimes\cH_{A_n})$ in Definition \ref{def:DEWn} that is a decomposable EW (DEW). There exist positive semi-definite Hermitian matrices $X_1, X_2,... \in \bbM_D$ such that      \begin{eqnarray}
\label{eq:H=X+YG}
        H = X_1 + \sum_{j>1} X_j^{\Gamma_{S_j}},
    \end{eqnarray}
where $\Gamma_{S_j}$ denotes the partial transpose with respect to the subsystem family $S_j\subset\{A_1,...,A_n\}$, such that $\abs{S_j}\le \lfloor{n \over 2}\rfloor$ and $S_j\cup S_k\ne \{A_1,...,A_n\}$. 

Then we use the mathematical induction. If a normalized DEW $H = X_1 + \sum_{j=1}^{k} X_j^{\Gamma_{S_j}}$ where 
\begin{eqnarray}
    H':=\frac{X_1 + \sum_{j=1}^{k-1} X_j^{\Gamma_{S_j}}}{\tr(X_1 + \sum_{j=1}^{k-1} X_j^{\Gamma_{S_j}})}
\end{eqnarray}
is also a normalized DEW. Then from the result before, we have
\begin{eqnarray}
    \tr[H'^2] \le 1.
\end{eqnarray}
At the same time time, since $X_k^{\Gamma_{S_k}}$ is a Hermitian matrix, we obtain
\begin{eqnarray}
    \tr[(\frac{X_k^{\Gamma_{S_k}}}{\tr(X_k^{\Gamma_{S_k}})})^2] =\frac{1}{\tr(X_k^{\Gamma_{S_k}})^2}\tr[(X_k^{\Gamma_{S_k}})^2] = 1.
\end{eqnarray}
One can see that
\begin{eqnarray}
    \tr(H')+\tr(X_k^{\Gamma_{S_k}})=\tr(H)=1.
\end{eqnarray}
Thus, using the convexity of the matrix function $\tr(H^2)$ in Lemma \ref{le:ConvexityOfTr(W^2)}, we obtain that
\begin{align}
    \tr(H^2)\le 1.
\end{align}
\end{proof}

The following result partially extends Lemma \ref{le:lambdaIN[-1/2,1)} (iii).
\begin{lemma}
\label{le:multi-DEW=lambda1<=1}
(i) The minimum eigenvalue of normalized multipartite EW is strictly upper bounded by zero.

(ii) The supremum of minimum eigenvalue of normalized multipartite DEW is also zero.
\end{lemma}

\begin{proof}
The minimum eigenvalue of a normalized multipartite EW $W$ is negative, so (i) holds. Let the DEW $W=\sum_j \l_j\proj{a_j}$ where $...>\l_k>0>\l_{k+1}>...$. Let $P=\sum_{s\ge k}\proj{a_s}$. If $q\ra0+$ then $W_1=qW+(1-q)P/\tr P$ is a normalized multipartite DEW, whose minimum eigenvalue is close to zero. We have proven claim (ii).  
\end{proof}

The following result partially extends Lemma \ref{le:lambdaIN[-1/2,1)} (v).

\begin{theorem}
\label{le:multi-DEW=lambda1<1}
We consider the eigenvalues of a normalized multipartite DEW. 

(i) The maximum eigenvalue is strictly upper bounded by one.

(ii) The supremum of maximum eigenvalue is also one. 

(iii) The infimum of maximum eigenvalue is $\frac{1}{D-1}$ where $D=d_1d_2 \cdots d_n$, and it is not attainable. 

(iv) The infimum of minimal eigenvalue is $-\frac{1}{2}$, and it is attainable.

(v) $\frac{1}{D-1}$ is a lower bound for $\lambda_1$ and $\lambda_1 > \frac{1}{D-1}$ when $\l_1$ is the maximum eigenvalue of any EW.
\end{theorem}

\begin{proof}
(i) We consider the DEW $W = P+ Q^{\G_A} + R^{\G_B} + S^{\G_C}$ where $\tr (W)=1$ and $P, Q, R, S \geq 0$, and $\G_X$ denotes the partial transpose w.r.t. system $X$. Then 
\begin{align}
\label{eq:l1(W)<=1}
    \lambda_1(W)=&\sup_{\|x\|=1} x^* W x \le\lambda_1(P) + \lambda_1(Q^{\G_A}) + \lambda_1(R^{\G_B}) + \lambda_1(S^{\G_C})\\
    =&\tr(P) \cdot \lambda_1\!\left(\frac{P}{\tr(P)}\right)+\tr(Q^{\G_A}) \cdot \lambda_1\!\left(\frac{Q^{\G_A}}{\tr(Q^{\G_A})}\right)\\
    +&\tr(R^{\G_B}) \cdot \lambda_1\!\left(\frac{R^{\G_B}}{\tr(R^{\G_B})}\right)+\tr(S^{\G_C}) \cdot \lambda_1\!\left(\frac{S^{\G_C}}{\tr(S^{\G_C})}\right) \\
    \le& \tr(P)+\tr(Q^{\G_A})+\tr(R^{\G_B})+\tr(S^{\G_C})=\tr(W)\le1.
\end{align}
Therefore, $\lambda_1(W)=1$ if and only if the equalities hold in all of the above inequalities. However, this implies that the principal eigenvector corresponding to the largest eigenvalue of $P$, $Q^{\Gamma_A}$, $R^{\Gamma_B}$, and $S^{\Gamma_C}$ is the same unit vector x, and $\lambda_1\!\left(\frac{P}{\tr(P)}\right)=\lambda_1\!\left(\frac{Q^{\G_A}}{\tr(Q^{\G_A})}\right)=\lambda_1\!\left(\frac{R^{\G_B}}{\tr(R^{\G_B})}\right)=\lambda_1\!\left(\frac{S^{\G_C}}{\tr(S^{\G_C})}\right)=1$, which also implies $P$, $Q^{\Gamma_A}$, $R^{\Gamma_B}$, and $S^{\Gamma_C}$ are all rank-one positive semi-definite matrices. Then it leads to a contradiction since $W$ must be an EW. So $\l_1(W)$ is strictly upper bounded by one.

We extend the above result to $n$-partite system. We assume 
\begin{eqnarray}
    W = P_0 + \sum_{k\ge1} P_k^{\Gamma_{S_k}},
\end{eqnarray}
where $S_k$'s are disjoint subsets of the whole set $A_1,...,A_n$, and each $P_k$ is a positive semi-definite matrix. Similar to \eqref{eq:l1(W)<=1}, we have 
\begin{align}
    \lambda_1(W) &= \sup_{\|x\|=1} x^* W x
    \le \lambda_1(P_0) + \sum_{k=1}^n \lambda_1(P_k^{\Gamma_{A_k}})\\ 
    &= \tr(P_0) \lambda_1\!\left(\frac{P_0}{\mathrm{Tr}P_0}\right) + \sum_{k=1}^n \tr(P_k^{\Gamma_{A_k}}) \lambda_1\!\left(\frac{P_k^{\Gamma_{A_k}}}{\mathrm{Tr}P_k^{\Gamma_{A_k}}}\right)\\
    &\le \tr(P_0) + \sum_{k=1}^n \tr(P_k^{\Gamma_{A_k}}) =\tr(W)\le 1.
\end{align}
Thus, $\lambda_1(W)=1$ if and only if the equalities hold in all of the above inequalities. However, this implies that $P_0$ and all $P_k^{\Gamma_{A_k}}$ are rank-one positive semi-definite matrices. It implies that the principal eigenvector corresponding to the largest eigenvalue of $P_0$ and $P_k^{\Gamma_{A_k}}$ for all $k=1,2,\cdots n$ is the same unit vector x, and $\lambda_1\!\left(\frac{P_0}{\tr(P_0)}\right)=\lambda_1\!\left(\frac{P_k^{\Gamma_{A_k}}}{\tr(P_k^{\Gamma_{A_k}})}\right)=1$, which also implies $P_0$, $P_k^{\Gamma_{A_k}}$ for all $k=1,2,\cdots n$ are all rank-one positive semi-definite matrices. Then it leads to a contradiction since $W$ must be an EW. So $\l_1(W)$ is strictly upper bounded by one.

(ii) Let $W$ be an $n$-partite normalized DEW, such that $\ket{a}$ is the eigenvalue of $\l_1(W)$. Then $W_1=pW+(1-p)\proj{a}$ is also an $n$-partite normalized EW, with $p\in(0,1)$. If $p\ra0$ then then $\l_1(W_1)$ approaches one. We have proven the claim.

(iii) First, let the maximum eigenvalue of a normalized multipartite DEW W be $\lambda_1$. We claim that $\frac{1}{D-1}$ is a lower bound for $\lambda_1$ and $\lambda_1 > \frac{1}{D-1}$. Indeed, suppose $\lambda_1 \le \frac{1}{D-1}$, then 
\begin{eqnarray}
    \sum_{i=1}^{D-1}\l_i \le (D-1) \cdot \l_1 = 1,
\end{eqnarray}
which implies that 
\begin{eqnarray}
    \l_D = \tr(W)-\sum_{i=1}^{D-1}\l_i \geq1-1=0. 
\end{eqnarray}
This contradicts the definition of EW. 

Secondly, we prove that $\lambda_1$ can be made arbitrarily close to $\frac{1}{D-1}$. We construct a DEW
\begin{eqnarray}
\label{eq:Wab}
    W = a\cdot \frac{I_D - |\alpha\rangle\langle\alpha|}{D-1}
+\,(1-a)\cdot \big|GHZ_n\big\rangle\langle GHZ_n|^{\,\Gamma_{A_1}}
\end{eqnarray}
where $a \in (0,1)$, $D=d_1d_2 \cdots d_n $ and $\alpha = |\alpha\rangle=\frac{1}{\sqrt{2}}|0\rangle_{A_1}\otimes|1\rangle^{\otimes(n-1)}-\frac{1}{\sqrt{2}}|1\rangle_{A_1}\otimes|0\rangle^{\otimes(n-1)}$. One can see that
\begin{align}
\big(|GHZ_n\rangle\langle GHZ_n|\big)^{\Gamma_{A_1}}
=\frac12\Big[
&|0\rangle_{A_1}\langle0|_{A_1}\otimes|0\rangle^{\otimes(n-1)}\langle0|^{\otimes(n-1)} \\
+&|1\rangle_{A_1}\langle0|_{A_1}\otimes|0\rangle^{\otimes(n-1)}\langle1|^{\otimes(n-1)} \\
+&|0\rangle_{A_1}\langle1|_{A_1}\otimes|1\rangle^{\otimes(n-1)}\langle0|^{\otimes(n-1)} \\
+&|1\rangle_{A_1}\langle1|_{A_1}\otimes|1\rangle^{\otimes(n-1)}\langle1|^{\otimes(n-1)}
\Big].
\end{align}
and it has four nonzero eigenvalues
\begin{eqnarray}
    \frac{1}{2}, \frac{1}{2}, \frac{1}{2}, -\frac{1}{2}. 
\end{eqnarray}
Besides, $|\alpha\rangle\langle\alpha|$ is a rank-one orthogonal projection. Therefore, it has eigenvalue $1$ with multiplicity one, with eigenvector $\alpha$, and eigenvalue $0$ with multiplicity $(D-1)$, whose eigenspace consists of all vectors orthogonal to $\alpha$. Thus we obtain the eigenvalue structure of $I_D - |\alpha\rangle\langle\alpha|$, and hence by simultaneous unitary diagonalization the largest eigenvalue of $W$ is 
\begin{eqnarray}
    \l_1=\frac{a}{D-1}+\frac{1-a}{2},
\end{eqnarray}
and the smallest is 
\begin{eqnarray}
    \l_D=-\frac{1-a}{2}.
\end{eqnarray}
This implies $W$ in \eqref{eq:Wab}, 
\begin{eqnarray}
    \l_1 \rightarrow \frac{1}{D-1} \ \text{as} \ a \rightarrow 1, \ \l_D \rightarrow -\frac{1}{2} \ \text{as} \ a \rightarrow 0. 
\end{eqnarray}
Therefore, the infimum of maximum eigenvalue is $\frac{1}{D-1}$ where $D=d_1d_2 \cdots d_n$, and it is not attainable.  

(iv) Using the construction in equation \eqref{eq:Wab} from the previous proof, we set $a = 0$ and obtain a DEW of the form given in Lemma \ref{le:GHZ_n}. This proves that the minimal eigenvalue $-\frac{1}{2}$ is attainable. Next, we claim that the minimal eigenvalue $\l_D$ satisfies
\begin{eqnarray}
    \l_D \geq -\frac{1}{2}. 
\end{eqnarray}
Indeed, for an arbitrary DEW $H$ of the form in Definition \ref{def:DEWn}, i.e., 
\begin{eqnarray}
    H = X_1 + \sum_{j>1} X_j^{\Gamma_{S_j}}
\end{eqnarray}
where $X_j \geq 0, j=1,2, \cdots$. Thus, it suffices to prove is that any mixed state $\sigma=\sum_i p_i |\psi_i\rangle\langle \psi_i|$ satisfies 
\begin{eqnarray}
    \lambda_D\big(\sigma^{\Gamma}\big) \geq -\frac{1}{2}
\end{eqnarray}
where $\sigma^{\Gamma}$ has the form of $\sum_i p_i |\psi_i\rangle\langle \psi_i|^{\Gamma_{A_i}}$ with $p_i \geq 0$ and $\sum_ip_i=1$. But from Lemma 4 in \cite{shen2020inertias}, we have known that for a partial transpose of a pure state $\r^{\G}$, we have $\l_{\text{min}}(\r^{\G}) \geq -\frac{1}{2}$. This allows us to prove that 
\begin{eqnarray}
\lambda_{\min}\big(\sigma^{\Gamma}\big)
\ge \sum_i p_i \left(-\frac{1}{2}\right)
= -\frac{1}{2}.
\end{eqnarray}
Therefore, for a normalized multipartite DEW, the infimum of minimal eigenvalue is $-\frac{1}{2}$, and it is attainable. 

(v) This result can be derived from the foregoing discussion in the beginning of (iii). 
\end{proof}

\subsection{Multipartite NDEWs}

\begin{theorem}
\label{le:TrW^2_NDEW}
    Let $W$ be a normalized $d_1 \times d_2 \times \cdots \times d_n$ EW. If $W$ is an NDEW, then the infimum of $\tr (W^2)$ is not attainable.
\end{theorem}

\begin{proof}
For normalized $d_1 \times d_2 \times \cdots \times d_n$ NDEWs, let 
\begin{eqnarray}
    \inf\{\tr (W^2)|W \ is \ an \ NDEW\}=k. 
\end{eqnarray}
Suppose that $k$ is attained by an NDEW $W_0$, then we consider 
\begin{eqnarray}
    W'=\frac{1}{1+\epsilon D}(W_0+\epsilon I_D)
\end{eqnarray}
where $D=d_1d_2 \cdots d ,\ \epsilon>0$. Using the properties of the trace and the normalization property of $W_0$, we obtain
\begin{eqnarray}
    \tr (W'^2) = \frac{1}{(1+\epsilon D)^2}[\tr (W_0^2)+2\epsilon \tr(W_0)+\epsilon^2D] \sim \frac{k+2\epsilon}{1+2\epsilon D}
\end{eqnarray}
where $\epsilon$ is sufficiently small. However, from the result of Theorem \ref{le:multiEWTRW^2}, we have $k>\frac{1}{D-1}>\frac{1}{D}$, thus 
\begin{eqnarray}
    \frac{k+2\epsilon}{1+2\epsilon D}<k,
\end{eqnarray}
which contradicts the definition of $k$. Therefore, if $W$ is an NDEW then the infimum of $\tr (W^2)$ is not attainable.
\end{proof}



\begin{remark}
\label{rem:AS}
We consider the state 
\begin{eqnarray}
\label{def:zeta2'}
    \zeta_2' := \frac{1}{D+n} \diag \{1+n, 1, \cdots , 1\} \in \bbM_D,
\end{eqnarray}
where $D=d_1...d_n$.
Then we have 
\begin{eqnarray}
\label{def:zeta2'^2}
    \zeta_2'^2 := \frac{1}{(D+n)^2} \diag \{(1+n)^2, 1, \cdots , 1\} \in \bbM_{D}. 
\end{eqnarray}
Using Lemma \ref{le:multi-AS}, we see that if $d_1=...=d_n$ and
\begin{eqnarray}
\label{ineq:d,n}
    \tr(\zeta_2'^2)=\frac{1}{(D+n)^2}[(1+n)^2+D-1] \le \frac{1}{D-2^{2-n}},
\end{eqnarray}
then \eqref{def:zeta2'} is a $d_1 \times d_1 \times \cdots \times d_1$ full-rank absolutely separable state. 

Nevertheless, we claim that the above conditions including \eqref{ineq:d,n} do not hold. We use an analytical approach, establish the conditions, and input them into Mathematica to obtain the results. It can be determined that within the range of the pair of positive integers $(d, n) \in [3,10000] \times [3,10000]$, there is no solutions that satisfies condition $d_1=...=d_n=d$ and \eqref{ineq:d,n}. 

More generally for integers $(d, n) \in [3,\infty) \times [3,\infty)$, $d^n - 2^{2-n}$ is monotonically increasing and remains greater than 0. In order to find a solution satisfying the above conditions, we transform it into solving the following binary inequality
\begin{eqnarray}
    \frac{(d^n - 2^{2-n})(d^n + (1+n)^2 - 1)}{(d^n + n)^2} \le 1, \ i.e.
\end{eqnarray}
\begin{eqnarray}
    d^n \le \frac{n^2 + 2^{2-n}(n^2 + 2n)}{n^2 - 2^{2-n}}. 
\end{eqnarray}
Thus, the solution must satisfy
\begin{eqnarray}
    2 \le d \le \left(\frac{n^2 + 2^{2-n}(n^2 + 2n)}{n^2 - 2^{2-n}}\right)^ {\frac{1}{n}}.
\end{eqnarray}
However, we have $\left(\frac{n^2 + 2^{2-n}(n^2 + 2n)}{n^2 - 2^{2-n}}\right)^ {\frac{1}{n}} \rightarrow 1$ as $n \rightarrow \infty$. Therefore, within the range of the pair of positive integers $(d, n) \in [3,\infty) \times [3,\infty)$, there is no solutions that satisfies condition $d_1=...=d_n=d$ and \eqref{ineq:d,n}. 
\qed
\end{remark}

This method was originally intended to prove absolute separability, and subsequently to establish other related results concerning EWs. However, as can be seen from Remark \ref{rem:AS}, it is better to adopt a different approach to obtain further information about EWs and NDEWs.


\section{Conclusion}
\label{sec:con}

We have shown separability and absolute-separability criterion of multipartite states, together with properties of block-positive (BP) matrices, optimality and system split and eigenvalues of multiparite EWs. Also, we have elucidated the interrelations among the eigenvalues of $2 \times n$ EWs. Then we have discussed other inequality governing eigenvalues, and further investigated the tightness of these inequalities with examples. Through an in-depth investigation of two-qubit states, we have attempted to generalize some properties of two-qubit states to arbitrary $n$-qubit states. All the results are comparable with the corresponding results for two-qubit systems, from which clear regularities can be observed. Some problems arising from this paper are as follows. First, we do not know the interval for the minimal eigenvalue of multipartite NDEW yet. Next, what is the relation between the optimality and eigenvalues of multipartite EWs? 
Besides, as far as bipartite EWs are concerned, one may study the inertia of $2\times4$ NDEW by comparing those of $2\times4$ DEW. Further, characterizing the eigenvalue sets of $2\times3$ EWs based on Johnston's paper is also an interesting problem.



\section*{Acknowledgments}
\label{sec:ack}	

Authors were supported by the NNSF of China (Grant No. 12471427), and the Fundamental Research Funds for the Central Universities (Grant Nos. ZG216S2110).

\appendix
\section{}
\label{app}

\subsection{The Proof of Lemma \ref{le:BPmn}}
\label{app:leBPmn}
\begin{proof}
Let $U$ be a unitary matrix such that $W_B$ satisfies 
\begin{eqnarray}
    U W_B U^* = \sum_{i}\ U W_{i,i} U^* = \sum_{i}\ x_i |e_i \rangle \langle e_i| = \left[\begin{matrix}x_1&&&&&\\&\ddots&&&&\\&&x_s&&&\\&&&0&&\\&&&&\ddots&\\&&&&&0\end{matrix}\right]_{n\times n}, 
\end{eqnarray}
where $s$ is the rank of $W_B$. Since each $W_{i,i}\geq 0$, $U W_{i,i} U^*$ has the form of $\left[\begin{matrix}W_{i,i}^\prime&O\\O&O\end{matrix}\right]$ where $W_{i,i}^\prime$ are all $s \times s$ matrices. Thus, 
\begin{eqnarray}
    (I_m \otimes U) W (I_m \otimes U^*) = \left[\begin{matrix}\left[\begin{matrix}W_{1,1}^\prime&O\\O&O\end{matrix}\right]&\left[\begin{matrix}W_{1,2}^\prime&O\\O&O\end{matrix}\right]&\cdots&\left[\begin{matrix}W_{1,m}^\prime&O\\O&O\end{matrix}\right]\\\left[\begin{matrix}W_{2,1}^\prime&O\\O&O\end{matrix}\right]&\left[\begin{matrix}W_{2,2}^\prime&O\\O&O\end{matrix}\right]&\cdots&\left[\begin{matrix}W_{2,m}^\prime&O\\O&O\end{matrix}\right]\\\vdots&\vdots&\ddots&\vdots\\\left[\begin{matrix}W_{m,1}^\prime&O\\O&O\end{matrix}\right]&\left[\begin{matrix}W_{m,2}^\prime&O\\O&O\end{matrix}\right]&\cdots&\left[\begin{matrix}W_{m,m}^\prime&O\\O&O\end{matrix}\right]\end{matrix}\right],
\end{eqnarray}
and we denote it by $\widetilde{W}$. Therefore, $\mathcal{R}(W) = (I_m \otimes U^*)\mathcal{R}(\widetilde{W}) \subset \mathbb{C}^m \otimes U^*(\mathbb{C}^s) = \mathbb{C}^m \otimes \mathcal{R}(W_B)$. Transforming back yields $\mathcal{R}(W) \subseteq \mathbb{C}^m \otimes \mathcal{R}(W_B)$. Combined with the fact that the row-support of W in the first factor lies in $\mathcal{R}(W_A)$, we obtain $\mathcal{R}(W) \subseteq \mathcal{R}(W_A) \otimes \mathcal{R}(W_B)$. 
\end{proof}

\subsection{The Proof of Lemma \ref{le:ConvexityOfTr(W^2)}}
\label{app:leConvexityOfTr(W^2)}
\begin{proof}
Let $A,B$ be two Hermitian matrices, then we claim that
\begin{eqnarray}
\label{eq:2AB<=A2+B2}
    2\tr(AB) \le \tr(A^2+B^2).
\end{eqnarray}
Indeed, 
\begin{eqnarray}
    \tr[(A-B)^2]=\tr(A^2-AB-BA+B^2)=\tr(A^2+B^2)-2\tr(AB).
\end{eqnarray}
Besides, since $A,B$ are two Hermitian matrices, $\tr[(A-B)^2]=\tr[(A-B)(A-B)^{\dagger}]\geq0$. Thus, $\tr(A^2+B^2)-2\tr(AB)\geq0$ and then we obtain \eqref{eq:2AB<=A2+B2}. Therefore for an arbitrary $\l \in [0,1]$, we have
\begin{align}
    \tr[(\l A + (1-\l)B)^2]=&\tr[\l ^2 A^2 + (1-\l)^2B^2]+2\l (1-\l)\tr(AB)\\
    \le & \tr[\l ^2 A^2 + (1-\l)^2B^2]+\l (1-\l)\tr(A^2+B^2) \label{ineq:convex}\\
    =&\tr[\l A^2+(1-\l)B^2]=\l \tr(A^2)+(1-\l)\tr(B^2).
\end{align}
So the proof is complete. 

In fact, this is a strictly convex function. If the two Hermitian matrices satisfy $A \neq B$, then we must have
\begin{eqnarray}
    \tr[(A-B)^2]>0, 
\end{eqnarray}
which implies that
\begin{eqnarray}
    \tr(A^2+B^2)>2\tr(AB),
\end{eqnarray}
and thus, when $\l < 1$, we replace the inequality sign in \eqref{ineq:convex} with a strict equality.
\end{proof}

\subsection{The Remark of Lemma \ref{lem:Existence}}
\label{app:lemExistence}

\begin{remark}
\label{rem:direction}
(i) Suppose that $\vec{x} \otimes \vec{y} \in \mathbb{C}^m \otimes \mathbb{C}^n$ where $\vec{x} \otimes \vec{y}$ is a nonzero product vector. If $\lambda \neq 0$, then $\vec{x} \otimes \vec{y} = \lambda\vec{x} \otimes \frac{1}{\lambda}\vec{y}$, which implies the product vector does not depend on scaling $\vec{x}, \ \vec{y}$ individually. Instead, it depends on their directions, i.e., rays of $\vec{x} \in \mathbb{C}^m$, $\vec{y} \in \mathbb{C}^n$. 

(ii) The set of all lines through the origin in $\mathbb{C}^k$, each of which has dimension one, is called the projective space $\mathbb{P}^{k-1}$.

(iii) A point $[\vec{v}] \in \mathbb{P}^{k-1} $ stands for  $\{ \lambda \vec{v} \mid \lambda \in \mathbb{C} \} $. Thus, for $\vec{x} \otimes \vec{y} \in \mathbb{C}^m \otimes \mathbb{C}^n$ where $\vec{x} \otimes \vec{y}$ is a nonzero product vector, its equivalence class $[\vec{x} \otimes \vec{y}]$ belongs to $ \mathbb{P}^{mn-1}$, and $[\vec{x} \otimes \vec{y}]$ only depends on $[\vec{x}] \in \mathbb{P}^{m-1}$, $[\vec{y}] \in \mathbb{P}^{n-1}$. 

(iv) $\sigma$ is called a Segre embedding, if
    \begin{align}
        \sigma : \mathbb{P}^{m-1} \times \mathbb{P}^{n-1} &\rightarrow\mathbb{P}^{mn-1}, \\
        ([\vec{x}], [\vec{y}]) &\mapsto [\vec{x} \otimes \vec{y}].
    \end{align}
    \qed
\end{remark}

\subsection{The Proof and The Example of Lemma \ref{lem:dew_eq}}
\label{app:lemdew_eq}
\begin{proof}
To begin with, we consider two sets
\begin{eqnarray}
    \mathcal{C} = \{ X + Y^{\Gamma} \mid X \ge 0,\ Y \ge 0 \} \subseteq \mathbb{M}_m \otimes \mathbb{M}_n,
\end{eqnarray}
and
\begin{eqnarray}
    \mathcal{P} = \{ \sigma \mid \sigma \ge 0,\ \sigma^{\Gamma} \ge 0 \} \subseteq \mathbb{M}_m \otimes \mathbb{M}_n. 
\end{eqnarray}
Since $\mathcal{C}$ is a closed convex cone, the bipolar theorem can be applied. We denote the dual cone of a cone $K$ by $K^*$, then $\mathcal{C} = \mathcal{C}^{**}$. To prove lemma \ref{lem:dew_eq}, it suffices to show that $\mathcal{C} = \mathcal{P}^*$. But $\mathcal{C} = \mathcal{C}^{**}$ tells us that this is equivalent to prove $\mathcal{C}^{**} = \mathcal{P}^*$, thus what we need to show is only $\mathcal{C}^* = \mathcal{P}$. 

First, we prove that $\mathcal{P} \subseteq \mathcal{C}^*$. Suppose $H \in \mathbb{M}_m \otimes \mathbb{M}_n$ is a DEW, i.e., $H = X + Y^{\Gamma}$ where $X, Y$ are positive semi-definite Hermitian matrices in $\mathbb{M}_m \otimes \mathbb{M}_n$. Then $\forall \text{ PPT state } \sigma $, we have
\begin{eqnarray}
    \tr H\sigma = \tr X\sigma + \tr Y^{\Gamma}\sigma.
\end{eqnarray}
Here, $\tr X\sigma \geq 0$ since $X$ is a positive semi-definite Hermitian matrix and $\sigma$ is a PPT state. On the other hand, we have
\begin{align}
    \tr Y^{\Gamma}\sigma &= \tr (\sum_{j,\ k=1}^{m}\left(e_{jk}\otimes\ I_n\right)Y\left(e_{jk}\otimes\ I_n\right)\sigma) \\
    &= \tr (\sum_{j,\ k=1}^{m}Y\left(e_{jk}\otimes\ I_n\right)\sigma \left(e_{jk}\otimes\ I_n\right)) = \tr Y\sigma^{\Gamma} \geq 0,
\end{align}
since $Y$ is a positive semi-definite Hermitian matrix and $\sigma^{\Gamma}$ is also a PPT state. Thus, we obtain that $\tr H\sigma \geq 0$. This implies $\mathcal{P} \subseteq \mathcal{C}^*$.

It remains to prove that $\mathcal{C}^* \subseteq \mathcal{P}$. We consider an arbitrary $W \in \mathcal{C}^*$ and claim that 
\begin{eqnarray}
\label{eq:W1>=0}
    W \geq 0,
\end{eqnarray}
and 
\begin{eqnarray}
\label{eq:W2>=0}
    W^\Gamma \geq 0. 
\end{eqnarray}
Indeed, suppose that \eqref{eq:W1>=0} dose not hold, then there must be a unit vector $\vec{v}$ s.t. $\vec{v}^* W \vec{v}<0$. Let $X=\vec{v}\vec{v}^*, \ Y=O$, then $H=X+Y^\G \in \mathcal{C}$ and $\tr (WH)=\tr (W\vec{v}\vec{v}^*)=\vec{v}^* W \vec{v}<0$, which leads to a contradiction. Similarly, suppose that \eqref{eq:W2>=0} dose not hold, then there must be a unit vector $\vec{v}$ s.t. $\vec{v}^* W^\G \vec{v}<0$. Let $X=O, \ Y=\vec{v}\vec{v}^*$, then $H=X+Y^\G \in \mathcal{C}$ and $\tr (WH)=\tr (WY^\G)=\tr(W^\G
Y)=\vec{v}^* W^\G \vec{v}<0$, which leads to a contradiction. So $\mathcal{C}^* \subseteq \mathcal{P}$.

To conclude, $\mathcal{C}^* = \mathcal{P}$. Therefore, an EW $H \in \mathbb{M}_m \otimes \mathbb{M}_n$ is a DEW if and only if $\tr(H\sigma) \ge 0, \quad \forall \text{ PPT state } \sigma $.
\end{proof}

\begin{example}
Here is a way of constructing a new EW by using a known EW. We consider a normalized DEW $W = X + Y^{\Gamma}$ where $X, Y$ are positive semi-definite Hermitian matrices in $\mathbb{M}_m \otimes \mathbb{M}_n$. Suppose it can be expressed as $\sum_{j=1}^{mn}\lambda_j^\downarrow |\psi_j\rangle \langle\psi_j|$ where $\langle\psi_j|\psi_j\rangle=1$. Then $1=\tr W=\tr X+\tr Y^{\Gamma}=\tr X+\tr Y$ where $X\geq0, \ Y>0$. Then we construct 
\begin{eqnarray}
    W_1 = \frac{s|\psi_1\rangle \langle\psi_1|+W}{s+1} =  \frac{s|\psi_1\rangle \langle\psi_1|+X}{s+1} + \frac{1}{s+1}Y^{\Gamma},
\end{eqnarray}
where $\tr W_1 = 1$, $W_1=W_1^*$, $(\vec{a}\otimes \vec{b})^* W_1 (\vec{a}\otimes \vec{b}) \geq 0, \quad \forall \vec{a}\otimes \vec{b}\in \mathbb{C}^m \otimes \mathbb{C}^n$, and $\langle\psi_{mn}|W_1|\psi_{mn}\rangle = \frac{\lambda_{mn}}{s+1}<0$. Therefore, $W_1$ is also a DEW.
\qed    
\end{example}

\subsection{The Proof of Theorem \ref{lem:finer_P}}
\label{app:lemfiner_P}
\begin{proof}
We define 
\begin{eqnarray}
\label{def:lambda}
    \l =\inf_{\r_1 \in D_{W_1}}\abs{\frac{\tr(W_2\r_1)}{\tr(W_1\r_1)}}.
\end{eqnarray}
Then we claim that $\l \geq1$. This is from case b below, while we shall prove three cases for further argument.

a. If $\tr(W_1\r)=0$, then $\tr(W_2\r)\le0$. Otherwise, when $\tr(W_2\r)>0$, we consider a $\r_1 \in D_{W_1}$ such that 
\begin{eqnarray}
    0 \le \r_1 + x\r \in D_{W_1}, \ \forall x\geq0.
\end{eqnarray}
However, if $x$ is large sufficiently, $\tr[W_2(\r_1 + x\r)]>0$, which imples $\r_1 + x\r \notin D_{W_2}$. This leads to a contradiction since $W_2$ is finer than $W_1$. Therefore, 
\begin{eqnarray}
\label{W1r=0}
    \tr(W_1\r)=0\ \Rightarrow \ \tr(W_2\r)\le0. 
\end{eqnarray}

b. If $\tr(W_1\r)<0$, then $\tr(W_2\r)\le \tr(W_1\r)$ because $\r + \abs{\tr(W_1\r)}I$ satisfies that 
\begin{eqnarray}
    \tr[W_1 (\r + \abs{\tr(W_1\r)}I)]=0
\end{eqnarray}
and then $\tr[W_2 (\r + \abs{\tr(W_1\r)}I)]\le 0$ from the preceding result. Therefore, 
\begin{eqnarray}
\label{W1r<0}
    \tr(W_1\r)<0 \ \Rightarrow \ \tr(W_2\r)\le \tr(W_1\r).
\end{eqnarray}
From case b, it is clear that $\l \geq1$. 

c. If $\tr(W_1\r)>0$, then $\l \tr(W_1\r) \geq \tr(W_2\r)$. This is because for $\r_1 \in D_{W_1}$, $\tr(W_1\r) \r_1 + \abs{\tr(W_1\r_1)}\r$ satisfies that 
\begin{eqnarray}
    \tr[W_1 (\tr(W_1\r) \r_1 + \abs{\tr(W_1\r_1)}\r)]=0
\end{eqnarray}
Then from the preceding result, we have 
\begin{eqnarray}
    \tr[W_2 (\tr(W_1\r) \r_1 + \abs{\tr(W_1\r_1)}\r)]\le0, i. e., 
\end{eqnarray}
\begin{eqnarray}
    \tr(W_1\r)\tr(W_2\r_1) + \abs{\tr(W_1\r_1)}\tr(W_2\r)\le 0,
\end{eqnarray}
which also implies that
\begin{eqnarray}
\label{eq:inf}
    \frac{\tr(W_2\r)}{\tr(W_1\r)} \le \frac{\abs{\tr(W_2\r_1)}}{\abs{\tr(W_1\r_1)}}.
\end{eqnarray}
Taking the infimum of the right hand of \eqref{eq:inf}, we have $\l \tr(W_1\r) \geq \tr(W_2\r)$. 
Therefore, 
\begin{eqnarray}
\label{W1r>0}
    \tr(W_1\r)>0 \ \Rightarrow \ \l \tr(W_1\r) \geq \tr(W_2\r).
\end{eqnarray}
We have proven case c.

Using the definition of $\l$ in \eqref{def:lambda}, we consider two cases (i) and (ii). 

(i) $\l =1$. Then we claim that $W_1=W_2$ where $k=0$. Indeed, for any $\r'=|a_1,\ a_2,\ \cdots a_n\rangle \langle a_1,\ a_2,\ \cdots a_n|$, we have 
\begin{eqnarray}
    \tr(W_1\r') \geq \tr(W_2\r').
\end{eqnarray}
However, $\tr(W_1)=\tr(W_2)=1$, then we obtain that 
\begin{eqnarray}
    \tr[(W_1-W_2)\r'] \geq 0.
\end{eqnarray}
Then for any given $\r \geq 0 $, we construct 
\begin{eqnarray}
    \r''=\r + M\cdot I,
\end{eqnarray}
where $I$ is an order-$d_1d_2 \cdots d_n$ identity matrix. When $M$ is large enough, $\r''$ is separable, and thus 
\begin{eqnarray}
    \tr(W_1 \r'') = \tr(W_2 \r''),
\end{eqnarray}
which implies $W_1 = W_2$. 

(ii) $\l \geq1$. We construct 
\begin{eqnarray}
    k=1-\frac{1}{\l}>0, \ P=\frac{1}{k}W_1-\frac{1-k}{k}W_2=\frac{1}{\l-1}(\l W_1 - W_2)
\end{eqnarray}
such that \eqref{eq:W1W2andP} holds. It suffices to prove that $P\geq 0$, i. e. 
\begin{eqnarray}
\label{ineq:P>=0}
    \tr(P\r) \geq 0, \ \forall\r.
\end{eqnarray}
Then we respectively discuss the three cases a, b and c above, in the following way. 

a. If $\tr(W_1\r)=0$, then $\tr(W_2\r)\le0$ from \eqref{W1r=0}. Then it is clear that \eqref{ineq:P>=0} must hold. 

b. If $\tr(W_1\r)<0$, then $\tr(W_2\r)\le \tr(W_1\r) < 0$ from \eqref{W1r<0}. Thus, 
\begin{align}
    \tr[(\l W_1 - W_2)\r]
    &=\tr(\inf_{\r_1 \in D_{W_1}}\abs{\frac{\tr(W_2\r_1)}{\tr(W_1\r_1)}} \cdot W_1\r) - \tr(W_2 \r)\\
    &\geq \tr(\abs{\frac{\tr(W_2\r)}{\tr(W_1\r)}} \cdot W_1\r) - \tr(W_2 \r)=0. 
\end{align}
Then \eqref{ineq:P>=0} must hold. 

c. If $\tr(W_1\r)>0$, then $\l \tr(W_1\r) \geq \tr(W_2\r)$ from \eqref{W1r>0}. Then \eqref{ineq:P>=0}  holds. 

To conclude, if $W_2$ is finer than $W_1$, then there is any positive operator $P$ in \eqref{ineq:P>=0} and an $k \in [0, 1)$ such that \eqref{eq:W1W2andP} holds. So the proof is complete. 
\end{proof}

\subsection{The Proof of Lemma \ref{le:n-partite EW split=EW}}
\label{app:len-partite EW split=EW}
\begin{proof}
(i) First, let the new family of disjoint subsystems be 
\begin{equation}
        S' = \{A_{1,1},A_{1,2}, \cdots A_{1,m_1}, \cdots  A_{n,m_n}\},
\end{equation}
and an arbitrary product vector of this family be 
\begin{equation}
    |\phi'\rangle = \bigotimes_{j=1}^n \left( \bigotimes_{k=1}^{m_j} |\phi_{j,k}\rangle \right).
\end{equation}
We can combine the terms within the brackets and define
\begin{eqnarray}
    |\phi_j\rangle = \bigotimes_{k=1}^{m_j} |\phi_{j,k}\rangle, 
\end{eqnarray}
then $|\phi_j\rangle$ is a valid state of the original system $A_j$. Therefore, $|\phi'\rangle$ can be rewritten as a pure product state under the original systems $A_1,\cdots.,A_n$ as follows,
\begin{eqnarray}
    |\phi'\rangle = |\phi_1\rangle \otimes |\phi_2\rangle \otimes \cdots \otimes |\phi_n\rangle.
\end{eqnarray}
By the definition of EW, it is non-negative on all pure product states of the original systems. Thus,
\begin{eqnarray}
    \langle \phi' | W' | \phi' \rangle = \langle \phi' | W | \phi' \rangle \geq 0.
\end{eqnarray}
Since $W$ is an EW, there exists an entangled state $|\psi \rangle$ in the original system such that $\langle \psi | W |\psi  \rangle < 0$. This state $|\psi \rangle$ also exists in the subdivided Hilbert space due to space isomorphism. Therefore, $W'$ still possesses a negative expectation value and is not a positive semi-definite operator. To conclude, the proof of (i) is complete. 

(ii) Since $W$ is a DEW, it can be decomposed as
\begin{eqnarray}
    W = X_1 + \sum_{j > 1} X_j^{\Gamma_{S_j}}
\end{eqnarray}
where $X_i \geq 0$ and $S_j$ satisfies $|S_j| \leq \lfloor \frac{n}{2} \rfloor$ and has no complement. When $A_j$ is split into $A_{j,1}, \cdots, A_{j,m_j}$, the set of original subsystems $S \subset \{A_1, \cdots, A_n\}$ corresponds to the new set of subsystems $S' = \bigcup_{A_j \in S} {A_{j,1}, \cdots, A_{j,m_j}}$. The partial transposition $\Gamma_S$ takes the transpose on the subsystems $A_j \in S$, which is equivalent to taking the partial transposition on all atomic subsystems in $S'$ in the new system, denoted as $\Gamma_{S'}$. DEW decomposition in the new system is 
\begin{eqnarray}
    W' = X_1' + \sum_{j > 1} (X_j')^{\Gamma_{S_j'}},
\end{eqnarray} 
where $X_i' \geq 0$. Originally $|S_j| \leq \lfloor \frac{n}{2} \rfloor$, and in the new system $|S_j'| = \sum_{A_k \in S_j} m_k$. The total number of parties is $N = \sum m_k$. Since $m_k \geq 1$, $|S_j'|$ is a part of N. The term $\lfloor \frac{N}{2} \rfloor$ in the definition is half of the total number of subsystems. Because $S_j$ is a subset of the original n blocks, $S_j'$ is the set of subsystems within these blocks, and its number will not exceed $\frac{N}{2}$. When $|S'_j|$ exceeds $\frac{N}{2}$, a global transposition can be applied so that the partial transposition becomes a system with order less than $\frac{N}{2}$. By definition, $W'$ remains a DEW. 
\end{proof}

\subsection{The proof of Lemma \ref{le:fullClassiEigen}}
\label{app:lefull}
\begin{proof}
Without loss of generality, we consider that $\G$ represents any partial transpose of the matrix $|v \rangle \langle v|$ with respect to the n-th subsystem, i.e.,  
\begin{eqnarray}
\label{eq:vvG}
    (|v \rangle \langle v|)^{\G} = \sum_{i,j=1}^r \alpha_i \alpha_j |a_{1,i}\rangle \langle a_{1,j}| \otimes |a_{2,i}\rangle \langle a_{2,j}| \otimes \cdots \otimes \overline{|a_{n,i}\rangle \langle a_{n,j}|}.
\end{eqnarray}

Firstly, let
\begin{eqnarray}
    |x_k \rangle = |a_{1,k}\rangle\otimes|a_{2,k}\rangle\otimes \cdots \otimes \overline{|a_{n,k}\rangle},
\end{eqnarray}
then using the orthogonality of the basis, we have
\begin{eqnarray}
    (|v \rangle \langle v|)^{\G}|x_k \rangle = \alpha_k^2 |x_k \rangle.
\end{eqnarray}

Secondly, let
\begin{eqnarray}
    |y_{k,l} \rangle ^{\pm} = \frac{1}{\sqrt{2}} (|a_{1,k}\rangle\otimes|a_{2,k}\rangle\otimes \cdots \otimes \overline{|a_{n,l}\rangle} \pm |a_{1,l}\rangle\otimes|a_{2,l}\rangle\otimes \cdots \otimes \overline{|a_{n,k}\rangle}).
\end{eqnarray}
Then using the orthogonality of the basis, the first term of the above expression is nonzero only when $i = l$ and $j = k$ and the second is nonzero only when $i = k$ and $j = l$ where indices i and j come from the previous equation \ref{eq:vvG}. Thus, 
\begin{eqnarray}
    (|v \rangle \langle v|)^{\G}|y_{k,l} \rangle ^{\pm} = \pm \alpha_k \alpha_l \cdot |y_{k,l} \rangle ^{\pm}.
\end{eqnarray}
Therefore, the proof is complete. 
\end{proof}

\subsection{The Proof of Lemma \ref{le:2^3}}
\label{app:le2^3}
\begin{proof}
Suppose the dimension of the negative eigenspace $V_-$ of W is at least 5. Then 
\begin{eqnarray}
    \dim \mathbb{P}V_- \geq 4.
\end{eqnarray}
According to the Segre variety intersection theorem, since $\text{Segre}(\mathbb{P}^1 \times \mathbb{P}^1 \times \mathbb{P}^1)$ has dimension three, we have
\begin{align}
    &\dim[\mathbb{P}V_- \cap \text{Segre}(\mathbb{P}^1 \times \mathbb{P}^1 \times \mathbb{P}^1)]\\
    =&\dim \mathbb{P}V_-+\dim[\text{Segre}(\mathbb{P}^1 \times \mathbb{P}^1 \times \mathbb{P}^1)]-\dim[\mathbb{P}V_- \cup \text{Segre}(\mathbb{P}^1 \times \mathbb{P}^1 \times \mathbb{P}^1)]\\
    \geq&\dim \mathbb{P}V_-+\dim[\text{Segre}(\mathbb{P}^1 \times \mathbb{P}^1 \times \mathbb{P}^1)]-\dim \mathbb{P}\mathcal{H}=4+3-7\geq 0,
\end{align}
which implies that
\begin{eqnarray}
    \mathbb{P}V_- \cap \text{Segre}(\mathbb{P}^1 \times \mathbb{P}^1 \times \mathbb{P}^1) \neq \varnothing.
\end{eqnarray}
The nonempty intersection means that there exists a product state $| a,b,c \rangle \in V_-$. Thus, there exists a negative eigenvalue $\lambda_- < 0$ such that
\begin{eqnarray}
    W | a,b,c \rangle = \lambda_- | a,b,c \rangle
\end{eqnarray}
with
\begin{eqnarray}
    \langle a,b,c | a,b,c \rangle =1.
\end{eqnarray}
Taking the inner product on both sides gives
\begin{eqnarray}
    \langle a,b,c | W | a,b,c \rangle = \lambda_- < 0.
\end{eqnarray}
However, it contradicts the definition of EW. Therefore, a 3-qubit witness $W \in \mathcal{H} = \mathbb{C}^2 \otimes \mathbb{C}^2 \otimes \mathbb{C}^2$ can have at most four linearly independent negative eigenvectors, and at most four negative eigenvalues.
\end{proof}

\subsection{The Proof of Lemma \ref{le:2^n}}
\label{app:le2^n}
\begin{proof}
Suppose the dimension of the negative eigenspace $V_-$ of W is at least $2^n-n$. Then 
\begin{eqnarray}
    \dim \mathbb{P}V_- \geq 2^n-n-1.
\end{eqnarray}
According to the Segre variety intersection theorem, since $\text{Segre}(\mathbb{P}^1 \times \mathbb{P}^1 \times \cdots  \times \mathbb{P}^1)$ has a dimension of at least n, we have
\begin{align}
    &\dim[\mathbb{P}V_- \cap \text{Segre}(\mathbb{P}^1 \times \mathbb{P}^1 \times \cdots  \times \mathbb{P}^1)]\\
    =&\dim \mathbb{P}V_-+\dim[\text{Segre}(\mathbb{P}^1 \times \mathbb{P}^1 \times \cdots  \times \mathbb{P}^1)]\\
    &-\dim[\mathbb{P}V_- \cup \text{Segre}(\mathbb{P}^1 \times \mathbb{P}^1 \times \cdots  \times \mathbb{P}^1)]\\
    \geq&\dim \mathbb{P}V_-+\dim[\text{Segre}(\mathbb{P}^1 \times \mathbb{P}^1 \times \cdots  \times \mathbb{P}^1)]-\dim \mathbb{P}\mathcal{H}\\
    =&(2^n-n-1)+n-(2^n-1)\geq 0,
\end{align}
which implies that
\begin{eqnarray}
    \mathbb{P}V_- \cap \text{Segre}(\mathbb{P}^1 \times \mathbb{P}^1 \times \cdots  \times \mathbb{P}^1) \neq \varnothing.
\end{eqnarray}
The nonempty intersection means that there exists a product state $| i_1,i_2,\cdots ,i_n \rangle \in V_-$. Thus, there exists a negative eigenvalue $\lambda_- < 0$ such that
\begin{eqnarray}
    W | i_1,i_2,\cdots ,i_n \rangle = \lambda_- | i_1,i_2,\cdots ,i_n \rangle
\end{eqnarray}
with
\begin{eqnarray}
    \langle i_1,i_2,\cdots ,i_n | i_1,i_2,\cdots ,i_n \rangle =1.
\end{eqnarray}
Taking the inner product on both sides gives
\begin{eqnarray}
    \langle i_1,i_2,\cdots ,i_n | W | i_1,i_2,\cdots ,i_n \rangle = \lambda_- < 0.
\end{eqnarray}
However, it contradicts the definition of EW. Therefore, a 3-qubit witness $W \in \mathcal{H} = \mathbb{C}^2 \otimes \mathbb{C}^2 \otimes \cdots \otimes \mathbb{C}^2$ can have at most four linearly independent negative eigenvectors, and at most four negative eigenvalues.
\end{proof}

\subsection{The Proof of Correlation Results for $2 \times n$ EWs}
\label{app:2tinEW}

The proof of Theorem \ref{le:2tin1234} is as follows. 

\begin{proof}
(i) If $\l_{3} < -\frac{1}{(2+2\sqrt{2})(2n-2)}$, then 
\begin{eqnarray}
    \sum_{i=3}^{2n} \l_{i}< (2n-2) \cdot \left[ -\frac{1}{(2+2\sqrt{2})(2n-2)} \right] <-\frac{1}{2+2\sqrt{2}},
\end{eqnarray}
which contradicts with Lemma \ref{le:EW=general} (iv). The same reasoning applies to the other results.

(ii) Setting $p = 2$ in the proof below yields the conclusion of (ii). 

(iii) Since $\rho \in \mathcal{AS}_{2,n}\left(\mathcal{AP}_{2,n}\right)$ if and only if $\lambda_1 \leq \lambda_{2n-1} + 2\sqrt{\lambda_{2n-2}\lambda_{2n}}$, we can construct a normalized $\rho \in \mathcal{AS}_{2,n}$ 
\begin{eqnarray}
\label{eq:asAAABCC}
    \r = \diag \{2p+1, \cdots ,2p+1,p^2,1,1\}
\end{eqnarray}
where $1 \le p \le 1+\sqrt{2}$. Note that the eigenvalues in \eqref{eq:asAAABCC} are arranged in non-increasing order. From the definition of EW, we obtain the positive semi-definiteness. Regarding the sorting operation corresponding to Von Neumann’s Trace Inequality, we obtain the following expression by multiplying and summing the eigenvalues in non-increasing and non-decreasing orders, respectively. We then have
\begin{eqnarray}
    (2p+1) \sum_{i=4}^{2n} \l_{i} + p^2 \l_3 + \left( 1-\l_3-\sum_{i=4}^{2n} \l_{i} \right)\geq0,
\end{eqnarray}
which implies that
\begin{eqnarray}
    2p\sum_{i=4}^{2n} \l_{i} + (p^2-1)\l_3 \geq -1. 
\end{eqnarray}

(iv) Since $\rho \in \mathcal{AS}_{2,n}\left(\mathcal{AP}_{2,n}\right)$ if and only if $\lambda_1 \leq \lambda_{2n-1} + 2\sqrt{\lambda_{2n-2}\lambda_{2n}}$, we can construct a normalized $\rho \in \mathcal{AS}_{2,n}$ 
\begin{eqnarray}
\label{eq:asAAABBC}
    \r = \diag \{m^2+2m, \cdots ,m^2+2m,m^2,m^2,1\}
\end{eqnarray}
where $m \geq 1$. Note that the eigenvalues in \eqref{eq:asAAABBC} are arranged in non-increasing order. From the definition of EW, we obtain the positive semi-definiteness. Regarding the sorting operation corresponding to Von Neumann’s Trace Inequality, we obtain the following expression by multiplying and summing the eigenvalues in non-increasing and non-decreasing orders, respectively. We then have
\begin{eqnarray}
    (m^2+2m) \sum_{i=4}^{2n} \l_{i} + m^2 (\l_3 + \l_2) + \left( 1-\l_2 -\l_3-\sum_{i=4}^{2n} \l_{i} \right)\geq0,
\end{eqnarray}
which implies that
\begin{eqnarray}
    (m^2+2m-1)\sum_{i=4}^{2n} \l_{i} + (m^2-1)(\l_3 + \l_2) \geq -1. 
\end{eqnarray}
\end{proof}
In particular, for (iii), when $p = 1+\sqrt{2}$ and $p = 1$, we obtain (b) and (c) in  Lemma \ref{le:EW=general} (iv), respectively. For (iv), when $m = 1$, we obtain (c) in  Lemma \ref{le:EW=general}. 

\begin{remark}
\label{re:constrOfW_SS}
Inequality \eqref{ineq:2n13>=} is tight. Indeed, using Theorem 3 of  \cite{Johnston_2018}, we construct an operator
\begin{eqnarray}
    W=\big(|\psi\rangle\langle\psi|\big)^\Gamma \oplus \mathbf{0}_{2(n-2)}
\end{eqnarray}
where $|\psi\rangle=\alpha_1|00\rangle+\alpha_2|11\rangle\in\mathbb{C}^2\otimes\mathbb{C}^2$ and $\alpha_1^2+\alpha_2^2=1$. Then W is a normalized EW with eigenvalues in non-increasing order
\begin{eqnarray}
    \lambda_1=\alpha_1^2,\quad \lambda_2=\alpha_1\alpha_2,\quad \lambda_3=\alpha_2^2,\quad \lambda_4=\dots=\lambda_{2n-1}=0,\quad \lambda_{2n}=-\alpha_1\alpha_2.
\end{eqnarray}
In this case, $\lambda_{2n}+\sqrt{\lambda_1\lambda_3}=-\alpha_1\alpha_2+\sqrt{\alpha_1^2\alpha_2^2}=0$. There exist legitimate examples for which equality holds, so Inequality \eqref{ineq:2n13>=} is tight.
\end{remark}

Let $S^\perp$ denote the orthogonal complement of subspace $S$ in $\mathbb{C}^2\otimes\mathbb{C}^n$, so that $\mathbb{C}^2\otimes\mathbb{C}^n = S\oplus S^\perp$. Then any operator can be written in block form as
\begin{eqnarray}
W=
\begin{bmatrix}
W_{SS} & W_{S S^\perp}\\
W_{S^\perp S} & W_{S^\perp S^\perp}
\end{bmatrix}
\end{eqnarray}
where $W_{S^\perp S}=(W_{S S^\perp})^\dagger$. 

The large space is split into two orthogonal regions $S$ and $S^\perp$ (the orthogonal complement of S). For the block notation of the operator 
\begin{eqnarray}
    W_{XY}
\end{eqnarray}
where the second subscript indicates which region the vector comes from and the first subscript indicates which region the vector is mapped to. For example, for $W_{S^\perp S}$, the vector comes from $S$ and, after mapping, goes to $S^\perp (S \to S^\perp)$.

The trivial example satisfies $W_{S S^\perp}=O,\ W_{S^\perp S^\perp}=O$, which corresponds to the direct-sum decomposition $W=W_{SS}\oplus \mathbf{0}$. Suppose that there exist a genuine $2\times n$ EW replacing the above two-qubit EW, then from the equality condition of Lemma \ref{le:2n13>=}, we must have
\begin{eqnarray}
    \lambda_{2n}=-\sqrt{\lambda_1\lambda_3}
\end{eqnarray}
and 
\begin{eqnarray}
    \lambda_1'=\lambda_1 ,\quad \lambda_3'=\lambda_3. 
\end{eqnarray}
This requires at least one of $W_{S^\perp S^\perp}\neq O$ or $W_{S S^\perp}\neq O$ to hold where $W_{SS}$ is a $4 \times 4$ matrix. We then analyze two separate cases.

(i) $W_{S S^\perp}\neq O$. Suppose $|v_1\rangle_S \in S$ satisfies 
\begin{eqnarray}
    W_{SS}|v_1\rangle_S=\lambda|_{W_{SS}}\cdot |v_1\rangle_S, 
\end{eqnarray}
then we construct
\begin{eqnarray}
    |v_1\rangle=\begin{bmatrix}|v_1\rangle_S\\ \boldsymbol{0}_{2n-4}\end{bmatrix}\in \mathcal{H}.
\end{eqnarray}
Theus we obtain 
\begin{eqnarray}
W|v_1\rangle
=\begin{bmatrix}
W_{SS} & W_{SS^\perp}\\
W_{S^\perp S} & W_{S^\perp S^\perp}
\end{bmatrix}
\begin{bmatrix}|v_1\rangle_S\\ \boldsymbol{0}\end{bmatrix}
=\begin{bmatrix}
W_{SS}|v_1\rangle_S \\
W_{S^\perp S}|v_1\rangle_S
\end{bmatrix}.
\end{eqnarray}
Substituting into the characteristic equation
$W_{SS} |v_1\rangle_S = \lambda |_{W_{SS}} |v_1\rangle_S$, we obtain
\begin{eqnarray}
W|v_1\rangle
=\begin{bmatrix}
\lambda|_{W_{SS}}\cdot |v_1\rangle_S \\
W_{S^\perp S}|v_1\rangle_S
\end{bmatrix}.
\end{eqnarray}
Since we assume $W_{S^\perp S} \neq 0$, there exists at least one $|v_1\rangle_S$ such that
\begin{eqnarray}
    W_{S^\perp S}|v_1\rangle_S \neq \mathbf{0}.
\end{eqnarray}
Hence the orthogonal complement component of $W|v_1\rangle$ is nonzero. Now using the definition of an eigenvector, if $|v_1\rangle$ is an eigenvector of W, then there exists a constant $\mu$ satisfying
\begin{eqnarray}
W|v_1\rangle = \mu |v_1\rangle =
\begin{bmatrix}
\mu |v_1\rangle_S \\
\mathbf{0}
\end{bmatrix},
\end{eqnarray}
which leads to a contradiction. 

(ii) From (i), we must have $W_{S S^\perp}= O$. Thus we have $W(S)\subseteq S,\quad W(S^\perp)\subseteq S^\perp$ where the subspace S is a nontrivial invariant subspace of W. By definition, a self-adjoint operator that possesses a nontrivial invariant subspace is called a reducible operator. In the above, the trivial example is
\begin{eqnarray}
    W = W_{SS} \oplus \mathbf 0.
\end{eqnarray}
The action of the whole operator can be decomposed into two independent subsystems that are disconnected and do not affect each other. The operator cannot map vectors in $S$ into $S^\perp$, nor can it map vectors in $S^\perp$ into $S$.

Therefore, since the partial trace $\tr_A$ is a linear map, we have
\begin{eqnarray}
    W_B=\mathrm{Tr}_A W=\mathrm{Tr}_A W_{SS}\;\oplus\; \mathrm{Tr}_A W_{S^\perp S^\perp}.
\end{eqnarray}
Let
\begin{eqnarray}
\Omega_1=\mathrm{Tr}_A W_{SS}\in\mathcal{B}(V_2),\quad
\Omega_2=\mathrm{Tr}_A W_{S^\perp S^\perp}\in\mathcal{B}(V_\perp), 
\end{eqnarray}
then under the decomposition $\mathbb{C}^n_B=V_2\oplus V_\perp$, 
\begin{eqnarray}
    W_B=\Omega_1\oplus\Omega_2.
\end{eqnarray}

\subsection{Explicit Construction of The Correlation Matrix for Multipartite GHZ}
\label{app:ExplicitConstructionGHZ}
\begin{eqnarray}
    |\Psi_3 \rangle \langle\Psi_3 |=
    \begin{bmatrix}\begin{bmatrix}\begin{bmatrix}\frac{1}{3} & 0 & 0 \\ 0 & 0 & 0 \\ 0 & 0 & 0\end{bmatrix} & O & O \\ O & O & O \\ O & O & O\end{bmatrix} & \begin{bmatrix}O & \begin{bmatrix}0 & \frac{1}{3} & 0 \\ 0 & 0 & 0 \\ 0 & 0 & 0\end{bmatrix} & O \\ \mathrm{O} & O & O \\ O & O & O\end{bmatrix} & \begin{bmatrix}O & O & \begin{bmatrix}0 & 0 & \frac{1}{3} \\ 0 & 0 & 0 \\ 0 & 0 & 0\end{bmatrix} \\ O & O & O \\ \mathrm{O} & O & O\end{bmatrix} \\ \begin{bmatrix}O & \mathrm{O} & O \\ \begin{bmatrix}0 & 0 & 0 \\ \frac{1}{3} & 0 & 0 \\ 0 & 0 & 0\end{bmatrix} & O & O \\ O & O & O\end{bmatrix} & \begin{bmatrix}O & O & O \\ O & \begin{bmatrix}0 & 0 & 0 \\ 0 & \frac{1}{3} & 0 \\ 0 & 0 & 0\end{bmatrix} & O \\ O & O & O\end{bmatrix} & \begin{bmatrix}O & O & O \\ O & O & \begin{bmatrix}0 & 0 & 0 \\ 0 & 0 & \frac{1}{3} \\ 0 & 0 & 0\end{bmatrix} \\ O & \mathrm{O} & O\end{bmatrix} \\ \begin{bmatrix}O & O & \mathrm{O} \\ O & O & O \\ \begin{bmatrix}0 & 0 & 0 \\ 0 & 0 & 0 \\ \frac{1}{3} & 0 & 0\end{bmatrix} & O & O\end{bmatrix} & \begin{bmatrix}O & O & O \\ O & O & \mathrm{O} \\ O & \begin{bmatrix}0 & 0 & 0 \\ 0 & 0 & 0 \\ 0 & \frac{1}{3} & 0\end{bmatrix} & O\end{bmatrix} & \begin{bmatrix}O & O & O \\ O & O & O \\ O & O & \begin{bmatrix}0 & 0 & 0 \\ 0 & 0 & 0 \\ 0 & 0 & \frac{1}{3}\end{bmatrix}\end{bmatrix}\end{bmatrix},
\end{eqnarray}
and
\begin{eqnarray}
\label{eq:Psi3G}
    |\Psi_3 \rangle \langle\Psi_3 |^\G=
    \begin{bmatrix}\begin{bmatrix}\begin{bmatrix}\frac{1}{3} & 0 & 0 \\ 0 & 0 & 0 \\ 0 & 0 & 0\end{bmatrix} & O & O \\ O & O & O \\ O & O & O\end{bmatrix} & \begin{bmatrix}O & O & O \\ \begin{bmatrix}0 & 0 & 0 \\ \frac{1}{3} & 0 & 0 \\ 0 & 0 & 0\end{bmatrix} & O & O \\ O & O & O\end{bmatrix} & \begin{bmatrix}O & O & O \\ O & O & O \\ \begin{bmatrix}0 & 0 & 0 \\ 0 & 0 & 0 \\ \frac{1}{3} & 0 & 0\end{bmatrix} & O & O\end{bmatrix} \\ \begin{bmatrix}O & \begin{bmatrix}0 & \frac{1}{3} & 0 \\ 0 & 0 & 0 \\ 0 & 0 & 0\end{bmatrix} & O \\ O & O & O \\ O & O & O\end{bmatrix} & \begin{bmatrix}O & O & O \\ O & \begin{bmatrix}0 & 0 & 0 \\ 0 & \frac{1}{3} & 0 \\ 0 & 0 & 0\end{bmatrix} & O \\ O & O & O\end{bmatrix} & \begin{bmatrix}O & O & O \\ O & O & O \\ O & \begin{bmatrix}0 & 0 & 0 \\ 0 & 0 & 0 \\ 0 & \frac{1}{3} & 0\end{bmatrix} & O\end{bmatrix} \\ \begin{bmatrix}O & O & \begin{bmatrix}0 & 0 & \frac{1}{3} \\ 0 & 0 & 0 \\ 0 & 0 & 0\end{bmatrix} \\ O & O & O \\ O & O & O\end{bmatrix} & \begin{bmatrix}O & O & O \\ O & O & \begin{bmatrix}0 & 0 & 0 \\ 0 & 0 & \frac{1}{3} \\ 0 & 0 & 0\end{bmatrix} \\ O & O & O\end{bmatrix} & \begin{bmatrix}O & O & O \\ O & O & O \\ O & O & \begin{bmatrix}0 & 0 & 0 \\ 0 & 0 & 0 \\ 0 & 0 & \frac{1}{3}\end{bmatrix}\end{bmatrix}\end{bmatrix}.
\end{eqnarray}
Therefore, we obtain all eigenvalues of $|\Psi_3 \rangle \langle\Psi_3 |^ \G$, i.e., $\frac{1}{3}$ and $-\frac{1}{3}$, with algebraic multiplicities six and three, respectively.

\subsection{Generalization of The Case where $\tr(W^2) \le 1$}
\label{app:UseConvexity}
We consider a normalized DEW $W=P+Q^{\G_A}+R^{\G_B}$ where $\frac{P+Q^{\G_A}}{\tr(P+Q^{\G_A})}$ is also a normalized DEW. Then from the result before, we have
\begin{eqnarray}
    \tr[(\frac{P+Q^{\G_A}}{\tr(P+Q^{\G_A})})^2] \le 1.
\end{eqnarray}
At the same time time, since $R^{\G_B}$ is a Hermitian matrix, we obtain
\begin{eqnarray}
    \tr[(\frac{R^{\G_B}}{\tr(R^{\G_B})})^2] =\frac{1}{\tr(R^{\G_B})^2}\tr[(R^{\G_B})^2] = 1.
\end{eqnarray}
One can see that
\begin{eqnarray}
    \tr(P+Q^{\G_A})+\tr(R^{\G_B})=\tr(P+Q^{\G_A}+R^{\G_B})=1.
\end{eqnarray}
Thus, using the convexity of the matrix function $\tr(W^2)$ in Lemma \ref{le:ConvexityOfTr(W^2)},
\begin{align}
    \tr(W^2)=&\tr[(P+Q^{\G_A}+R^{\G_B})^2]\\
    =&\tr\left[\left( \tr(P+Q^{\G_A}) \cdot \frac{P+Q^{\G_A}}{\tr(P+Q^{\G_A})} + \tr(R^{\G_B}) \cdot \frac{R^{\G_B}}{\tr(R^{\G_B})} \right)^2\right]\\
    \le&\tr(P+Q^{\G_A}) \cdot \tr[(\frac{P+Q^{\G_A}}{\tr(P+Q^{\G_A})})^2] + \tr(R^{\G_B}) \cdot \tr[(\frac{R^{\G_B}}{\tr(R^{\G_B})})^2]\\
    \le&\tr(P+Q^{\G_A})+\tr(R^{\G_B})=1.
\end{align}

\bibliographystyle{unsrt}

\bibliography{nalan=ew} 

@article{johnston2018inverse,
 title={The inverse eigenvalue problem for entanglement witnesses},
 author={Johnston, N. and Patterson, E.},
 journal={Linear Algebra Appl},
 volume={550},
 pages={1--27},
 year={2018},
 publisher={Elsevier}
}

@article{serrano2024absolute,
 title={Absolute-separability witnesses for symmetric multiqubit states},
 author={Serrano-Ens{\'a}stiga, E. and Denis, J. and Martin, J.},
 journal={Phys. Rev. A},
 volume={109},
 number={2},
 pages={022430},
 year={2024},
 publisher={APS}
}

@article{chruscinski2025mirroredn,
	title={A mirrored pair of optimal non-decomposable entanglement witnesses for two qudits does exist},
	author={Chru{\'s}ci{\'n}ski, D. and Bera, A. and Bae, J. and Hiesmayr, B. C},
	journal={Sci. Rep.},
	volume={15},
	number={1},
	pages={28205},
	year={2025},
	publisher={Nature Publishing Group UK London}
}

@article{bourennane2004experimental,
	title={Experimental detection of multipartite entanglement using witness operators},
	author={Bourennane, M. and Eibl, M. and Kurtsiefer, C. and Gaertner, S. and Weinfurter, H. and G{\"u}hne, O. and Hyllus, P. and Bru{\ss}, D. and Lewenstein, M. and Sanpera, A.},
	journal={Phys. Rev. Lett.},
	volume={92},
	number={8},
	pages={087902},
	year={2004},
	publisher={APS}
}

@article{bera2023structure,
	title={On the structure of mirrored operators obtained from optimal entanglement witnesses},
	author={Bera, A. and Bae, J. and Hiesmayr, B. C and Chru{\'s}ci{\'n}ski, D.},
	journal={Sci. Rep.},
	volume={13},
	number={1},
	pages={10733},
	year={2023},
	publisher={Nature Publishing Group UK London}
}

@article{bae2020mirrored,
	title={Mirrored entanglement witnesses},
	author={Bae, J. and Chru{\'s}ci{\'n}ski, D. and Hiesmayr, B. C},
	journal={Npj Quantum Inf.},
	volume={6},
	number={1},
	pages={15},
	year={2020},
	publisher={Nature Publishing Group UK London}
}

@article{Chru2014Entanglement,
	title={Entanglement witnesses: construction, analysis and classification},
	author={Chrucinski, D. and  Sarbicki, G. },
	journal={J. Phys. A},
	volume={47},
	number={48},
	pages={483001},
	year={2014},
}

@article{lewenstein2000optimization1,
	title={Optimization of entanglement witnesses},
	author={Lewenstein, M. and Kraus, B. and Cirac, J. I. and Horodecki, P.},
	journal={Phys. Rev. A},
	volume={62},
	number={5},
	pages={052310},
	year={2000},
	publisher={APS}
}

@article{guhne2006nonlinear,
	title={Nonlinear entanglement witnesses},
	author={G{\"u}hne, O. and L{\"u}tkenhaus, N.},
	journal={Phys. Rev. Lett.},
	volume={96},
	number={17},
	pages={170502},
	year={2006},
	publisher={APS}
}

@article{hyllus2005relations,
	title={Relations between entanglement witnesses and Bell inequalities},
	author={Hyllus, P. and G{\"u}hne, O. and Bru{\ss}, D. and Lewenstein, M.},
	journal={Phys. Rev. A},
	volume={72},
	number={1},
	pages={012321},
	year={2005},
	publisher={APS}
}

@article{chruscinski2009spectral,
	title={Spectral conditions for entanglement witnesses versus bound entanglement},
	author={Chru{\'s}ci{\'n}ski, D. and Kossakowski, A. and Sarbicki, G.},
	journal={Phys. Rev. A},
	volume={80},
	number={4},
	pages={042314},
	year={2009},
	publisher={APS}
}

@article{feng2024inertia,
	title={Inertia of two-qutrit entanglement witnesses},
	author={Feng, C. C. and Chen, L. and Xu, C. and Shen, Y.},
	journal={	Linear Multilinear A},
	volume={72},
	number={3},
	pages={451--473},
	year={2024},
	publisher={Taylor \& Francis}
}

@article{shen2020inertias,
	title={Inertias of entanglement witnesses},
	author={Shen, Y. and Chen, L. and Zhao, L. J.},
	journal={J. Phys. A},
	volume={53},
	number={48},
	pages={485302},
	year={2020},
	publisher={IOP Publishing}
}

@article{2025Spectral,
  title={Spectral characterizations of entanglement witnesses},
  author={ Song, Zhiwei  and  Chen, Lin },
  year={2025},
}

@article{Ransford_2026,
  title={A 98-qubit trapped-ion quantum computer with all-to-all connectivity},
  volume={655},
  ISSN={1476-4687},
  url={http://dx.doi.org/10.1038/s41586-026-10676-4},
  DOI={10.1038/s41586-026-10676-4},
  number={8121},
  journal={Nature},
  publisher={Springer Science and Business Media LLC},
  author={Ransford, Anthony and Allman, M. S. and Arkinstall, Jake and others},
  year={2026},
  month={6},
  pages={81--86}
}

@article{slater2009eigenvalues,
	title={Eigenvalues, separability and absolute separability of two-qubit states},
	author={Slater, P. B},
	journal={J. Geom. Phys.},
	volume={59},
	number={1},
	pages={17--31},
	year={2009},
	publisher={Elsevier}
}

@article{han2017separability,
	title={Separability of three qubit},
	author={Han, K. H. and Kye, S. H.},
	journal={J. Phys. A},
	volume={50},
	number={14},
	pages={145303},
	year={2017},
	publisher={IOP Publishing}
}

@article{tanaka2014determining,
	title={Determining eigenvalues of a density matrix with minimal information in a single experimental setting},
	author={Tanaka, T. and Ota, Y. and Kanazawa, M. and Kimura, G. and Nakazato, H. and Nori, F.},
	journal={Phys. Rev. A},
	volume={89},
	number={1},
	pages={012117},
	year={2014},
	publisher={APS}
}

@book{nielsen2010quantum,
	title={Quantum computation and quantum information},
	author={Nielsen, M. A and Chuang, I. L},
	year={2010},
	publisher={Cambridge university press}
}

@article{horodecki1997separability,
	title={Separability criterion and inseparable mixed states with positive partial transposition},
	author={Horodecki, P.},
	journal={Phys. Lett. A},
	volume={232},
	number={5},
	pages={333--339},
	year={1997},
	publisher={Elsevier}
}

@article{horodecki1996necessary,
	title={On the necessary and sufficient conditions for separability of mixed quantum states},
	author={Horodecki, M. and Horodecki, P. and Horodecki, R.},
	journal={Phys. Lett. A},
	volume={223},
	number={1},
	year={1996},
	publisher={Citeseer}
}

@article{peres1996separability,
	title={Separability criterion for density matrices},
	author={Peres, A.},
	journal={Phys. Rev. Lett.},
	volume={77},
	number={8},
	pages={1413},
	year={1996},
	publisher={APS}
}

@article{leinaas2010numerical,
	title={Numerical studies of entangled positive-partial-transpose states in composite quantum systems},
	author={Leinaas, J. M. and Myrheim, J. and Sollid, P. O.},
	journal={Phys. Rev. A},
	volume={81},
	number={6},
	pages={062329},
	year={2010},
	publisher={APS}
}

@article{leinaas2007extreme,
	title={Extreme points of the set of density matrices with positive partial transpose},
	author={Leinaas, J. M. and Myrheim, J. and Ovrum, E.},
	journal={Phys. Rev. A},
	volume={76},
	number={3},
	pages={034304},
	year={2007},
	publisher={APS}
}

@article{ganguly2014witness,
	title={Witness of mixed separable states useful for entanglement creation},
	author={Ganguly, N. and Chatterjee, J. and Majumdar, A. S.},
	journal={Phys. Rev. A},
	volume={89},
	number={5},
	pages={052304},
	year={2014},
	publisher={APS}
}

@article{lewenstein2001characterization,
	title={Characterization of separable states and entanglement witnesses},
	author={Lewenstein, M. and Kraus, B and Horodecki, P and Cirac, J.},
	journal={Phys. Rev. A},
	volume={63},
	number={4},
	pages={044304},
	year={2001},
	publisher={APS}
}

@article{sarbicki2008spectral,
	title={Spectral properties of entanglement witnesses},
	author={Sarbicki, G.},
	journal={J. Phys. A: Math. Theor.},
	volume={41},
	number={37},
	pages={375303},
	year={2008},
	publisher={IOP Publishing}
}

@article{3x3inertia2022changchun,
author = {Changchun Feng and Lin Chen and Chang Xu and Yi Shen},
title = {Inertia of two-qutrit entanglement witnesses},
journal = {Linear and Multilinear Algebra},
volume = {0},
number = {0},
pages = {1-23},
year  = {2022},
publisher = {Taylor and Francis},
doi = {10.1080/03081087.2022.2159304},
}

@article{2020Inertias,
  title={Inertias of entanglement witnesses},
  author={ Shen, Y.  and  Chen, L.  and  Zhao, L. J. },
  journal={Journal of Physics A: Mathematical and Theoretical},
  volume={53},
  number={48},
  pages={485302 (27pp)},
  year={2020},
}

@article{2018Optimization,
  title={Optimization of ultrafine entanglement witnesses},
  author={ Shen, Shu Qian  and  Xu, Ti Run  and  Fei, Shao Ming  and  Li-Jost, Xianqing  and  Li, Ming },
  journal={physical review a},
  volume={97},
  number={3},
  year={2018},
}

@article{2019Design,
	title={Design and Experimental Performance of Local Entanglement Witness Operators},
	author={ Amaro, D.  and M Müller},
	year={2019},
	}

@article{2020Measurement,
	title={Measurement-Device-Independent Entanglement Witness of Tripartite Entangled States and Its Applications},
	author={ Li, Zheng Da  and  Zhao, Qi  and  Zhang, Rui  and  Liu, Li Zheng  and  Pan, Jian Wei },
	journal={Physical Review Letters},
	volume={124},
	number={16},
	year={2020},
	}

@article{2013Sperling,
	title={Multipartite Entanglement Witnesses},
	author={J. Sperling and W. Vogel },
	journal={Physical Review Letters},
	doi={https://doi.org/10.1103/PhysRevLett.111.110503},
	year={2013},
	}

@article{2012CM,
	title={Witnessing Quantum Coherence: from solid-state to biological systems},
	author={ C.-M. Li and N. Lambert and Y.-N. Chen G.-Y. Chen and F. Nori},
	journal={Scientific Reports},
	volume={2},
	number={1},
	pages={885},
	doi={https://doi.org/10.1038/srep00885},
	year={2012},
	}

@article{2003LGH,
	title={ Separable balls around the maximally mixed multipartite quantum states.},
	author={Leonid Gurvits and Howard Barnum.},
	journal={Phys. Rev. A},
	volume={68},
	number={},
	pages={042312},
	year={2003},
	}

@article{2025Sufficient,
  title={Sufficient criteria for absolute separability in arbitrary dimensions via linear map inverses},
  author={ Abellanetvidal, Jofre  and  Sanpera, Anna },
  journal={IOP Publishing Ltd},
  year={2025},
}

@article{Johnston_2018,
   title={The inverse eigenvalue problem for entanglement witnesses},
   volume={550},
   ISSN={0024-3795},
   url={http://dx.doi.org/10.1016/j.laa.2018.03.043},
   DOI={10.1016/j.laa.2018.03.043},
   journal={Linear Algebra and its Applications},
   publisher={Elsevier BV},
   author={Johnston, Nathaniel and Patterson, Everett},
   year={2018},
   month=Aug, pages={1–27} }

@article{2006FGS,
	title={Witnessed entanglement.},
	author={ FERNANDO G. S. L. BRANDÃO and REINALDO O. VIANNA.},
	journal={International Journal of Quantum Information},
	volume={04},
	number={02},
	pages={331–340},
	doi={},
	year={2006},
	}

@article{2006OGN,
	title={Nonlinear entanglement witnesses.},
	author={Otfried Gühne and Norbert Lütkenhaus.},
	journal={Phys. Rev. Lett.},
	volume={96},
	number={},
	pages={170502},
	doi={},
	year={2006},
	}

@article{2011JDB,
	title={ Device-independent witnesses of genuine multipartite entanglement.},
	author={Jean-Daniel Bancal and Nicolas Gisin and Yeong-Cherng Liang and Stefano Pironio.},
	journal={Phys. Rev. Lett.},
	volume={106},
	number={},
	pages={250404},
	doi={},
	year={2011},
	}

@article{Pezze2017MultipartiteEI,
  author = {Pezz{\`e}, Luca and Gabbrielli, Marco and Lepori, Luca and Smerzi, Augusto},
  title = {Multipartite Entanglement in Topological Quantum Phases},
  journal = {Physical Review Letters},
  volume = {119},
  number = {25},
  pages = {250401},
  year = {2017},
  month = {Dec},
  doi = {10.1103/PhysRevLett.119.250401},
  url = {https://doi.org/10.1103/PhysRevLett.119.250401},
  issn = {1079-7114},
  archivePrefix = {arXiv},
  eprint = {1706.06539},
  primaryClass = {quant-ph}
}

@article{Igloi2023Entanglement,
  author = {Iglói, Ferenc and Tóth, Géza},
  title = {Entanglement witnesses in the {XY} chain: Thermal equilibrium and postquench nonequilibrium states},
  journal = {Physical Review Research},
  volume = {5},
  number = {1},
  pages = {013158},
  year = {2023},
  month = {Mar},
  doi = {10.1103/PhysRevResearch.5.013158},
  url = {https://doi.org/10.1103/PhysRevResearch.5.013158},
  archivePrefix = {arXiv},
  eprint = {2207.04842},
  primaryClass = {quant-ph}
}

@article{MacLean2018Direct,
  author = {MacLean, Jean-Philippe W. and Donohue, John M. and Resch, Kevin J.},
  title = {Direct Characterization of Ultrafast Energy-Time Entangled Photon Pairs},
  journal = {Physical Review Letters},
  volume = {120},
  number = {5},
  pages = {053601},
  year = {2018},
  month = {Jan},
  publisher = {American Physical Society},
  doi = {10.1103/PhysRevLett.120.053601},
  url = {https://doi.org/10.1103/PhysRevLett.120.053601},
  archivePrefix = {arXiv},
  eprint = {1710.11541},
  primaryClass = {quant-ph}
}

@article{Zhang2018Characterization,
  author = {Zhang, Yu-Ran and Zeng, Yu and Fan, Heng and You, J. Q. and Nori, Franco},
  title = {Characterization of Topological States via Dual Multipartite Entanglement},
  journal = {Physical Review Letters},
  volume = {120},
  number = {25},
  pages = {250501},
  year = {2018},
  month = {Jun},
  publisher = {American Physical Society},
  doi = {10.1103/PhysRevLett.120.250501},
  url = {https://doi.org/10.1103/PhysRevLett.120.250501},
  archivePrefix = {arXiv},
  eprint = {1712.05286},
  primaryClass = {quant-ph}
}

\end{document}